\documentclass[11pt]{llncs}
\usepackage{t1enc}
\usepackage{times}
\usepackage{epsfig}
\usepackage{url}
\usepackage{amsmath}
\usepackage{amssymb}
\usepackage{fullpage}
\usepackage{pdflscape}
\usepackage{latexsym}
\usepackage{todonotes}
\usepackage{multicol}
\usepackage{wrapfig}
\usepackage{numprint}
\usepackage{theorem}
\usepackage{makeidx}
\usepackage{enumitem}
\usepackage{hyperref}
\usepackage{relsize}
\usepackage[sort&compress, sectionbib]{natbib}
\usepackage{graphicx}
\usepackage{wrapfig}

\newcommand{\Oh}[1]{\mathcal{O}\!\left( #1\right)}
\newcommand{\Is}       {:=}

\newif\ifDoubleBlind
\DoubleBlindfalse

\ifDoubleBlind
\usepackage[left] {lineno}
\definecolor {infocolor} {rgb} {0.6,0.6,0.6}

\linenumbers
\fi{}
\newcommand{\ie}{i.e.\ }
\newcommand{\etal}{et~al.\ }

\usepackage{tikz,pgfplots}
\usepgfplotslibrary{colorbrewer}
\usetikzlibrary{patterns}
\npdecimalsign{.} 

\usepackage[labelformat=simple]{subcaption}

\usepackage[utf8]{inputenc}

\usepackage[algo2e,linesnumbered,ruled,lined]{algorithm2e}
\usepackage{algorithmic}

\newcommand{\CC}{C\texttt{++}}

\pgfplotscreateplotcyclelist{black white 2}{%
	every mark/.append style={solid},mark=*\\%
	every mark/.append style={solid},mark=square*\\%
	every mark/.append style={solid},mark=o\\%
	mark=star\\%
	every mark/.append style={solid},mark=diamond*\\%
	densely dashed,every mark/.append style={solid},mark=*\\%
	densely dashed,every mark/.append style={solid},mark=square*\\%
	densely dashed,every mark/.append style={solid},mark=o\\%
	densely dashed,every mark/.append style={solid},mark=star\\%
	densely dashed,every mark/.append style={solid},mark=diamond*\\%
}

\pgfplotsset{
  cycle list/Dark2,
  every axis/.append style={
    ylabel near ticks,
    log basis y={2},
    log basis x={2},
    legend cell align={left},
    legend style={font=\Large},
    label style={font=\Large},
    title style={font=\Large},
    tick label style={font=\Large},
    cycle multiindex* list={
      Dark2
      \nextlist
      black white 2
      \nextlist
   },
  }
}

\let\doendproof\endproof
\renewcommand\endproof{~\hfill$\square$\doendproof}

\def\MdN{\ensuremath{\mathbb{N}}}

\newcommand{\algname}{VieCut}

\usepackage{color}

\newcommand{\mytitle}{Practical Minimum Cut Algorithms}
\begin{document}
\title{\mytitle}
\ifDoubleBlind
\author{}
\institute{}
\else
\author{
Monika Henzinger\inst{1}, Alexander Noe\inst{1}, Christian Schulz\inst{2} and Darren Strash\inst{3}
} \institute{University Vienna, Vienna, Austria\\ \email{\{monika.henzinger,
alexander.noe\}@univie.ac.at}\\ \and  
Karlsruhe Institute of Technology, Karlsruhe, Germany and University of Vienna, Vienna, Austria \\
\email{christian.schulz@\{kit.edu, univie.ac.at\}}
        \and
        Department of Computer
Science, Colgate University, Hamilton, NY, USA\\ \email{dstrash@cs.colgate.edu}}
\fi{}
\date{today}
 \pagestyle{plain}

 \maketitle
 \thispagestyle{plain}
\begin{abstract} The minimum cut problem for an undirected edge-weighted graph
asks us to divide its set of nodes into two blocks while minimizing the weight sum of the cut edges. Here, we
introduce a linear-time algorithm to compute near-minimum cuts. Our
algorithm is based on cluster contraction using label propagation and Padberg
and Rinaldi's contraction heuristics [SIAM Review, 1991]. We give both
sequential and shared-memory parallel implementations of our algorithm.
Extensive experiments on both real-world and 
generated instances show that our algorithm finds the optimal cut
on nearly all instances significantly faster than other state-of-the-art algorithms while our error rate is lower than that of other heuristic algorithms. In addition, our parallel algorithm shows good scalability. 
\end{abstract}
\ifDoubleBlind
\vfill
\pagebreak
\setcounter{page}{1}
\else
\fi{}

\section{Introduction} Given an undirected graph with non-negative edge weights,
the \emph{minimum cut problem} is to partition the vertices into two sets so
that the sum of edge weights between the two sets is minimized. A minimum cut is
often also referred to as the \textit{edge connectivity} of a
graph~\cite{nagamochi1992computing,henzinger2017local}. The problem has
applications in many fields. In particular, for network
reliability~\cite{karger2001randomized,ramanathan1987counting}, 
assuming equal failure chance on edges, the smallest edge cut in the network has
the highest chance to disconnect the network; in VLSI
design~\cite{krishnamurthy1984improved}, to minimize the number of
connections between microprocessor blocks; and as a subproblem in the
branch-and-cut algorithm for solving the Traveling Salesman Problem and other
combinatorial problems~\cite{padberg1991branch}.

For minimum cut algorithms to be viable for these (and other) applications they must be fast on small data sets---and scale to large data sets. Thus, an algorithm should have either linear or near-linear running time, or have an efficient parallelization. Note that \emph{all} existing exact algorithms have non-linear running time~\cite{hao1992faster,henzinger2017local,karger1996new}, where the fastest of these is the deterministic algorithm of Henzinger~\etal\cite{henzinger2017local} with running time $\Oh{m \log^2{n} \log \log^2 n}$. Although this is arguably near-linear theoretical running time, it is not known how the algorithm performs in practice. Even the randomized algorithm of Karger and Stein~\cite{karger1996new} which finds a minimum cut only with high probability, requires $\Oh{n^2\log^3{n}}$ time, later improved by Karger~\cite{karger2000minimum} to $
\Oh{m\log^3{n}}$. There is a linear time approximation algorithm, namely the $(2 + \varepsilon)$-approximation algorithm by Matula~\cite{matula1993linear}. However, the quality of Matula's algorithm in practice is currently unknown---no experiments have been published, although Chekuri \etal provide an implementation~\cite{Chekuri:1997:ESM:314161.314315,code}.

To the best of our knowledge, there exists only one parallel algorithm for the minimum cut problem: Karger and Stein present a parallel variant for their random contraction algorithm~\cite{karger1996new} which computes a minimum cut with high probability in polylogarithmic time on $n^2$ processors. However, we are unaware of any implementation of this algorithm, or any variant that may be suitable for running in a shared-memory setting.  Furthermore, no experiments have been published with any parallelizations. This is not altogether surprising, as previous experimental studies of the minimum cut problem only include instances of up to \numprint{85900} vertices and \numprint{1024000} edges~\cite[Table 4.2, Table A.22]{levine1997experimental}, which can be solved in milliseconds even without parallelization.

\subsubsection{Our Results.}
In this paper, we give the first practical \emph{shared-memory parallel} algorithm for the minimum cut problem. Our algorithm is heuristic, randomized and has running time $\Oh{n+m}$ when run sequentially. The algorithm works in a multilevel fashion: we repeatedly reduce the input graph size with both heuristic and exact techniques, and then solve the smaller remaining problem with exact methods.
Our heuristic technique identifies edges that are unlikely to be in a minimum cut using label propagation
introduced by Raghavan~\etal\cite{raghavan2007near} and contracts them in bulk.
We further combine this technique with the exact reduction routines
from Padberg and Rinaldi~\cite{padberg1990efficient}.
We perform extensive experiments comparing our algorithm with other heuristic algorithms as well as exact algorithms on real-world and generated instances, which include graphs on up to \numprint{70} million vertices and \numprint{5} billion edges---the largest graphs ever used for experiments for the minimum cut problem.
Results indicate that our algorithm finds optimal cuts
on almost all instances and also that the empirically observed error rate is lower than competing
heuristic algorithms that come with guarantees on the solution quality. At the same time, even when run sequentially, our algorithm is significantly faster (up to a factor of $4.85$) than other state-of-the-art algorithms.

\section{Related Work}
We review algorithms for the global minimum cut and related problems. A closely related problem is the \textit{minimum s-t-cut} problem, which asks
for a minimum cut with nodes $s$ and $t$ in different partitions. Ford and
Fulkerson~\cite{ford1956maximal} proved that minimum $s$-$t$-cut 
is equal to maximum $s$-$t$-flow. Gomory and Hu~\cite{gomory1961multi}
observed that the (global) minimum cut can be computed with $n-1$ minimum
$s$-$t$-cut computations.
For the following decades, this result by Gomory and Hu was used to find better
algorithms for global minimum cut using improved maximum flow
algorithms~\cite{karger1996new}. One of the fastest known maximum flow
algorithms is the push-relabel algorithm~\cite{goldberg1988new} by Goldberg and
Tarjan.

Hao and Orlin~\cite{hao1992faster} adapt the push-relabel algorithm to pass
information to future flow computations. When a push-relabel iteration is finished,
they implicitly merge the source and sink to form a new sink and find a new
source. Vertex heights are maintained over multiple iterations of push-relabel. With
these techniques they achieve a total running time of $O(mn\log{\frac{n^2}{m}})$ for
a graph with $n$ vertices and $m$ edges, which is asymptotically equal to a
single run of the push-relabel algorithm.

Padberg and Rinaldi~\cite{padberg1990efficient} give a set of heuristics for
edge contraction. Chekuri~\etal\cite{Chekuri:1997:ESM:314161.314315} give an
implementation of these heuristics that can be performed in time linear in the
graph size. Using these heuristics it is possible to sparsify a graph while
preserving at least one minimum cut in the graph. If their algorithm does not
find an edge to contract, it performs a maximum flow computation, giving the
algorithm worst case running time $O(n^4)$. However, the heuristics can also be
used to improve the expected running time of other algorithms by applying them
on interim graphs~\cite{Chekuri:1997:ESM:314161.314315}.

Nagamochi \etal\cite{nagamochi1992computing,nagamochi1994implementing} give a minimum
cut algorithm which does not use any flow computations. Instead, their
algorithm uses maximum spanning forests to find a non-empty set of contractible
edges. This contraction algorithm is run until the graph is contracted into a single
node. The algorithm has a running time of $O(mn+n^2\log{n})$. Wagner and
Stoer~\cite{stoer1997simple} give a simpler variant of the algorithm of
Nagamochi, Ono and Ibaraki~\cite{nagamochi1994implementing}, which has a
the same asymptotic time complexity. The performance of this algorithm on
real-world instances, however, is significantly worse than the performance of the
algorithms of Nagamochi, Ono and Ibaraki or Hao and
Orlin, as shown in experiments conducted by J\"unger \etal\cite{junger2000practical}.
Both the algorithms of Hao and Orlin or Nagamochi, Ono and Ibaraki achieve close to linear running time on most benchmark instances~\cite{junger2000practical,Chekuri:1997:ESM:314161.314315}. There are no
parallel implementation of either algorithm known to us.

Kawarabayashi and Thorup~\cite{kawarabayashi2015deterministic} give a
deterministic near-linear time algorithm for the minimum cut problem, which runs
in $O(m \log^{12}{n})$. Their algorithm works by growing contractible regions
using a variant of PageRank~\cite{page1999pagerank}. It was later
improved by Henzinger~\etal\cite{henzinger2017local} to run in $O(m
\log^2{n} \log \log^2 n)$ time.

Based on the algorithm of
Nagamochi, Ono and Ibaraki, Matula~\cite{matula1993linear} gives a
$(2+\varepsilon)$-approximation algorithm for the minimum cut problem. The algorithm contracts more edges than the algorithm of
Nagamochi, Ono and Ibaraki to guarantee a linear time complexity while still
guaranteeing a $(2+\varepsilon)$-approximation~factor. Karger and
Stein~\cite{karger1996new} give a randomized Monte Carlo algorithm based on random edge contractions. This algorithm returns the minimum cut with high probability and a larger cut otherwise.


\section{Preliminaries}\label{s:preliminaries}

\subsubsection*{Basic Concepts.}
Let $G = (V, E, c)$ be a weighted undirected graph with vertex set $V$, edge set $E \subset V \times V$ and
non-negative edge weights $c: E \rightarrow \MdN$. 
We extend $c$ to a set of edges $E' \subseteq E$ by summing the weights of the edges; that is, $c(E')\Is \sum_{e=\{u,v\}\in E'}c(u,v)$. We apply the same notation for single nodes and sets of nodes.
Let $n = |V|$ be the
number of vertices and $m = |E|$ be the number of edges in $G$. The \emph{neighborhood}
$N(v)$ of a vertex $v$ is the set of vertices adjacent to $v$. The \emph{weighted degree} of a vertex is the sum of the weight of its incident edges. For brevity, we simply call this the \emph{degree} of the vertex.
For a set of vertices $A\subseteq V$, we denote by $E[A]\Is \{(u,v)\in E \mid u\in A, v\in V\setminus A\}$; that is, the set of edges in $E$ that start in $A$ and end in its complement.
A cut $(A, V
\setminus A)$ is a partitioning of the vertex set $V$ into two non-empty
\emph{partitions} $A$ and $V \setminus A$, each being called a \emph{side} of the cut. The \emph{capacity} of a cut $(A, V
\setminus A)$ is $c(A) = \sum_{(u,v) \in E[A]} c(u,v)$.
A \emph{minimum cut} is a cut $(A, V
\setminus A)$ that has smallest weight $c(A)$ among all cuts in $G$. We use $\lambda(G)$ (or simply
$\lambda$, when its meaning is clear) to denote the value of the minimum cut
over all $A \subset V$. At any point in the execution of a minimum cut algorithm,
$\hat\lambda(G)$ (or simply $\hat\lambda$) denotes the lowest upper bound of the
minimum cut that an algorithm discovered until that point. 
For a vertex $u \in V$ with minimum vertex degree, the size of the \emph{trivial cut} $(\{u\}, V\setminus \{u\})$ is equal to the vertex degree of $u$.
Hence, the minimum vertex degree $\delta(G)$ can serve as initial~bound.

When \emph{clustering} a graph, we are looking for \emph{blocks} of nodes $V_1$,\ldots,$V_k$ 
that partition $V$, i.e., $V_1\cup\cdots\cup V_k=V$ and $V_i\cap V_j=\emptyset$
for $i\neq j$. The parameter $k$ is usually not given in advance.
Many algorithms tackling the minimum cut problem use \emph{graph contraction}.
Given
an edge $(u, v) \in E$, we define $G/(u, v)$ to be the graph after \emph{contracting
edge} $(u, v)$. In the contracted graph, we delete vertex $v$ and all edges
incident to this vertex. For each edge $(v, w) \in E$, we add an edge $(u, w)$
with $c(u, w) = c(v, w)$ to $G$ or, if the edge already exists, we give it the edge
weight $c(u,w) + c(v,w)$.

\subsubsection*{The Minimum Cut Algorithm of Nagamochi, Ono and Ibaraki.}
\label{sec:noi}
We discuss the algorithms by Nagamochi, Ono and Ibaraki~\cite{nagamochi1992computing,nagamochi1994implementing}
in greater detail since our work makes use of the tools proposed by those authors.
The minimum cut algorithm of Nagamochi, Ono and Ibaraki works on graphs with positive integer weights and computes a minimum cut
by building \emph{edge-disjoint maximum spanning forests} and contracting all edges that are not in one of the $\hat\lambda - 1$ first spanning forests. There is at least one minimum cut that contains no contracted edges. The algorithm has worst case running time
$\Oh{mn+n^2\log n}$. In experimental evaluations \cite{Chekuri:1997:ESM:314161.314315,junger2000practical} it is one of the fastest minimum cut algorithms on real-world
instances.

The algorithm uses a modified \emph{breadth-first graph traversal} (BFS) algorithm~\cite{nagamochi1992computing,nagamochi1994implementing} to
find edges that can be contracted without increasing the minimum cut. The algorithm starts at an arbitrary vertex.
In each step, the algorithm then visits the vertex that is most strongly connected to the already visited vertices.

Using the modified BFS routine, the algorithm computes a lower bound $q(e)$ for
each edge $e = (v, w)$ for the smallest cut $\lambda(G, v, w)$, which places $v$
and $w$ on different sides of the cut. After finishing the BFS, all edges with
$q(e) \geq \hat\lambda$ are contracted as there exists at least one minimum cut
without these edges. Afterwards, the algorithm continues on the contracted
graph. A single iteration of the subroutine can be performed in $\Oh{m+n \log
n}$. The authors show that in each BFS run, at least one edge of the graph can
be contracted~\cite{nagamochi1992computing}.
This yields a total running time of $O(mn+n^2 \log n)$. 
However, in practice the number of iterations is typically much less than $n-1$, often scaling proportional to  $\log n$.

\section{\texttt{\algname}: A Parallel Heuristic Minimum-Cut Algorithm}
\label{sec:pr}
In this section we explain our new approach to the minimum cut problem. Our
algorithm is based on edge contractions: we find densely connected vertices
in the graph and contract those into single vertices.  Due to the way
contractions are defined, we ensure that a minimum cut of the contracted graph
corresponds to a minimum cut of the input graph. Once the graph is contracted,
we apply exact reductions. These two contraction steps are repeated until the graph
has a constant number of vertices. We apply an exact minimum cut algorithm to
find the optimal cut in~the~contracted~graph.

Throughout our algorithm we maintain a variable $\hat\lambda$, which denotes the
current lowest upper bound for the minimum cut. In the beginning, $\hat\lambda$
equals the minimum node degree of $G$. After every contraction, if the minimum
node degree in the contracted graph is smaller than $\hat\lambda$, we set
$\hat\lambda$ to the minimum node degree of the contracted graph. As we only
perform contractions and thus do not introduce any new cuts¸ we can guarantee
that our algorithm will never output a value that is lower than the minimum cut.

The rest of this section is organized as follows: first we give a general overview of our algorithm. Then we introduce
the label propagation algorithm~\cite{raghavan2007near}, which we use to find
clusters in the input graph. Note that cluster contraction is an aggressive
contraction technique. In contrast to previous approaches, it enables us to
drastically shrink the size of the input graph. The intuition behind this
technique is that each cluster consists of a set of nodes that all belong to the
same side of the cut; as there are usually many edges inside the clusters and only a few between clusters. Thus we contract each cluster.  We continue this section with
the description of the Padberg-Rinaldi heuristics~\cite{padberg1990efficient}. These are a
set of conditions that can be used to find edges that can be contracted without
increasing the value of the minimum cut. We use these contractions to further
contract the graph after the label propagation step.
Both contractions are repeated until the graph has a constant number of vertices. Then, we apply an exact minimum cut algorithm. We finish the section
with a discussion of our parallelization.

\subsection{Fast Minimum Cuts}
The algorithm of Karger and Stein spends a large amount of time computing graph
contractions recursively.  One idea to speed up their algorithm therefore, is to
increase the number of contracted edges per level. However, this strategy is
undesirable: it increases the error both in theory and in practice, as their
algorithm selects edges for contraction at random.  We solve this
problem by introducing an aggressive coarsening strategy that contracts a large
number of edges that are unlikely to be in a minimum cut.

We first give a high level overview before diving into the details of the
algorithm.  Our algorithm starts by using a label propagation algorithm to
cluster the vertices into densely connected clusters.  We then
use a correcting algorithm to find misplaced vertices that should form a singleton
cluster.  Finally, we contract the graph and apply the exact reductions of Padberg and
Rinaldi~\cite{padberg1990efficient}. We repeat these contraction steps
until the graph has at most a constant number $n_0$ of vertices.
When the contraction step is finished we apply the algorithm of Nagamochi, Ono and
Ibaraki~\cite{nagamochi1994implementing} to find the minimum cut of the
contracted graph. Finally, we transfer the resulting cut into a cut in the
original graph. Pseudocode can be found in Appendix~\ref{app:algorithm}.

The \emph{label propagation algorithm} (LPA) was proposed by Raghavan
\etal\cite{raghavan2007near} for graph clustering. It is a fast algorithm that
locally minimizes the number of edges cut. We outline the algorithm briefly.
Initially, each node is in its own cluster/block, \ie the initial block ID of a
node is set to its node ID.  The algorithm then works in rounds. In each round,
the nodes of the graph are traversed in a random order.  When a node $v$ is
visited, it is \emph{moved} to the block that has the strongest connection to
$v$, \ie it is moved to the cluster $C$ that maximizes $c(\{(v, u) \mid u \in
N(v) \cap C \})$. Ties are broken uniformly at random. The block IDs of round~$i$ are used as initial block IDs of round $i+1$. In the original
formulation~\cite{raghavan2007near}, the process is repeated until the process
converges and no vertices change their labels in a round.
Kothapalli \etal\cite{kothapalli2013analysis} show that label propagation finds all
clusters in few iterations with high probability, when the graph has a distinct cluster structure.
Hence, we perform at most $\ell$ iterations of the algorithm, where $\ell$ is a
tuning parameter.  One LPA round can be implemented to run in $\Oh{n+m}$
time. As we only perform $\ell$ iterations, the algorithm runs
in $\Oh{n+m}$ time as long as $\ell$ is a constant. In this formulation the algorithm has no bound on the number of clusters. However, we can modify the first iteration of the algorithm, so that
a vertex $i$ is not allowed to change its label when another vertex already
moved to block $i$. In a connected graph this guarantees that each cluster has at least two vertices and the contracted graph has at most $\frac{|V|}{2}$ vertices. 
The only exception are
connected components consisting of only a single vertex, \ie \emph{isolated} vertices
with a degree of $0$, which can not be contracted by the label propagation algorithm. However, when such a vertex is detected, our minimum cut algorithm terminates immediately.
In practice we do not use the modification, as the label propagation usually returns far fewer than $\frac{|V|}{2}$ clusters. 

Once we have computed the clustering with label propagation, we search for
single misplaced vertices using a \emph{correcting algorithm}. A misplaced vertex is a vertex, whose removal from its cluster improves the minimum weighted degree of the contracted graph. Figure~\ref{fig:misplaced} gives an example in which the clustering
misplaces a vertex. To find misplaced vertices, we
sweep over all vertices and check for each vertex whether it is misplaced.
We only perform this correcting algorithm on small clusters, which have a size of up to $\log_2(n)$ vertices, as it is likely that large clusters would have more than a single node misplaced at a time. In general, one can enhance this algorithm by starting at any node whose removal would lower the cluster degree and greedily adding neighbors whose removal further lowers the remaining cluster degree. However, even when performing this greedy search on all clusters, this did not yield further improvement over the single vertex version on small clusters.

After we computed the final clustering, we \emph{contract it} to obtain a
coarser graph.  Contracting the clustering works as follows: each block of the
clustering is contracted into a single node.
There is an edge
between two nodes $u$ and $v$ in the contracted graph if the two corresponding
blocks in the clustering are adjacent to each other in $G$, \ie block $u$ and
block $v$ are connected by at least one edge.  The weight of an edge $(A,B)$ is
set to the sum of the weight of edges that run between block $A$ and block $B$
of the clustering.  Our contractions ensure that a minimum cut of the
coarse graph corresponds to a cut of the finer graph with the same
value, but not vice versa:
we can not guarantee that a minimum cut of the
contracted graph is equal to a minimum cut of the original graph. It is possible
that a single cluster contains nodes from both sides of the cut. In this case,
contracting the cluster eliminates this minimum cut. If all minimum cuts are
eliminated, $\lambda(G_\mathcal{C}) > \lambda(G)$. Thus our newly introduced reduction for the minimum
cut problem is \emph{inexact}. However,~the~following~lemma~holds (the proof can be found in Appendix~\ref{proof:lemma}):

\begin{figure}[t]
\centering
\includegraphics[width=6cm]{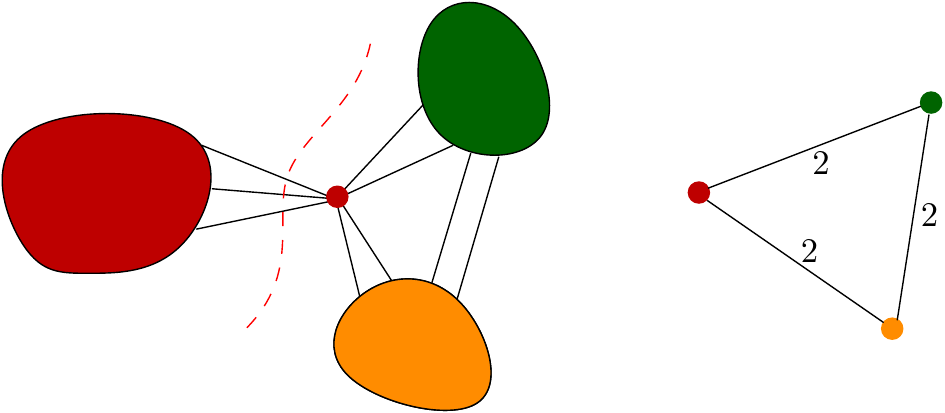}\hspace*{.5cm}
\includegraphics[width=6cm]{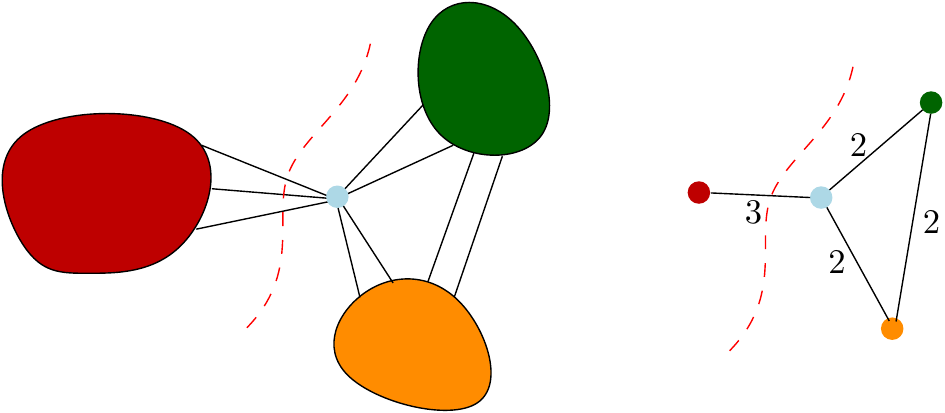}
\caption{A case in which label propagation misplaces vertices. Left: label propagation assigns the centered vertex correctly to the left (red) cluster. However, this results in a situation in which the contracted graph does not contain the minimum cut anymore. Setting the centered vertex to be a singleton fixes this problem.  }
\label{fig:misplaced}
\end{figure}

\begin{lemma} If there exist a minimum cut of $G$ such that each cluster of the clustering $\mathcal{C}$ is completely contained in one side of the minimum
  cut of $G$ and $|V_\mathcal{C}| > 1$, then $\lambda(G) = \lambda(G_\mathcal{C})$. \label{lemma:cluster}
\end{lemma}

\paragraph*{Exact Reductions by Padberg and Rinaldi.}
We use the Padberg-Rinaldi reductions to further shrink the size of the graph. These are exact reductions, which do not modify the size of the minimum cut. Our algorithm contracts all edges which are marked by the Padberg-Rinaldi heuristics. In our experiments, we also tried run the exact reductions first and cluster contraction last. However, this resulted in a slower algorithm since not many exact reductions could be applied on the initial unweighted network.
We now briefly discuss the exact reductions by Padberg and Rinaldi~\cite{padberg1990efficient}:

\begin{lemma}[Padberg-Rinaldi]
  If two vertices $v,w \in V$ with an edge $(v,w) \in E$ satisfy at least one of the following four conditions and $(v,w)$ is not the only edge adjacent to either $v$ or $w$, then they can be contracted without increasing the value of the minimum cut:
  \begin{enumerate}
  \item $c(v,w) \geq \hat\lambda$,
  \item $c(v) \leq 2 c(v, w)$ or $c(w) \leq 2 c(v, w)$,
  \item $\exists u \in V$ such that $c(v) \leq 2 \{c(v,w) + c(v,u)\}$ and $c(w)
\leq 2\{c(v,w)+c(w,u)\}$, or
  \item $c(v, w) + \sum_{u \in V}\min\{c(v,u), c(w,u)\} \geq \hat\lambda$.
  \end{enumerate}
  \label{lemma:pr14}
\end{lemma}

\ifDoubleBlind
\fi{}
We use this lemma to find contractible edges in the graph. In our
implementation, we perform so-called ``runs'' of the Padberg-Rinaldi heuristics
similar to the implementation of Chekuri \etal\cite{Chekuri:1997:ESM:314161.314315}, where each run takes linear time and is split into two passes. In the
first pass of a run, we iterate over all edges of $G$ and check conditions 1 and
2. Whenever we encounter an edge $(u,v)$, that satisfies either condition 1 or
2, we mark it as contractible. After
finishing the pass, we build the contracted graph. More precisely, we perform contraction in linear time by deleting all unmarked edges, contracting connected components and then re-adding the deleted edges as defined in the contraction process. In practice, we achieve better performance using a union-find data structure~\cite{galler1964improved}, which results in a running time of $\Oh{n\alpha(n)+m}$

It is not possible to perform an exhaustive check for conditions 3 and 4 on all
triangles in $G$ in linear time, as a graph might have as many as
$\Theta(m^{\frac{3}{2}})$ triangles~\cite{schank2005finding}. In the beginning
of a pass, we mark each vertex as unscanned. When scanning two adjacent
unscanned vertices $v$ and $w$, we check condition 3 for all vertices $u$ in the
common neighborhood $N(v) \cap N(w)$. As we iterate over all vertices in the
common neighborhood, we can compute the sum in condition 4 by adding up the
smaller of the two edge weights for each vertex in the common neighborhood.  Afterwards we mark both
$v$ and $w$ as scanned. This ensures a time complexity of $\Oh{n+m}$, as each edge is
processed at most twice. However, not for all possible edges $(v,w)$ it is tested whether the vertices $v$ and $w$ can be contracted.
\paragraph{Final Step: Exact Minimum Cut Algorithm.} To find the minimum cut of the final
problem kernel, we use the minimum cut algorithm of Nagamochi, Ono and Ibaraki,
as described in Section~\ref{sec:noi}.

\begin{lemma}
  The algorithm \texttt{\algname} has a running time complexity of $\Oh{n+m}$. (Proof in Appendix~\ref{proof:lemma})
\end{lemma}

\subsection{Parallelization}

We describe how to parallelize \texttt{\algname}. We parallelize each part of the algorithm, except the final invocation of the algorithm of Nagamochi, Ono, and Ibaraki.
\paragraph{Parallel Label Propagation.}
To perform the label update for vertex $v$, we only need to consider vertices in the neighborhood $N(v)$. Therefore the
label propagation algorithm can be implemented in parallel on shared-memory
machines~\cite{staudt2013engineering} using the \texttt{parallel for} directive from the
OpenMP~\cite{dagum1998openmp} API.
We store the cluster affiliation for all vertices in an
array of size $n$, where position $i$ denotes the cluster affiliation of vertex
$i$. We explicitly do not perform label updates in a \texttt{critical} section,
as each vertex is only traversed once and the race conditions are not critical
but instead introduce another source of randomness.

\paragraph{Parallel Correcting Step.}
As the clusters are independent of each other for this correcting step, we parallelize it a cluster level, \ie a cluster is checked by a single thread but each thread can check a different cluster without the need for locks or mutexes.
\paragraph{Parallel Graph Contraction.}
\label{contract}
After label propagation has partitioned the graph into $c$ clusters, we
build the cluster graph. As the time to build this contracted graph is
not negligible, we parallelize graph contraction as well. One of the $p$
threads performs the memory allocations to store the contracted graph, while
the other $p-1$ threads prepare the data for this contracted graph. When
$c^2 > n$, we parallelize the graph on a cluster
level. To build the contracted vertex
for cluster $C$, we iterate over all outgoing edges $e = (u,v)$ for all vertices
randomness and graph locality by randomly shuffling small blocks of vertex IDs
but traversing these shuffled blocks successively. If $v \in C$, we discard the edge, otherwise we add $c(u,v)$ to the
edge weight between $C$ and the cluster of vertex $v$.
When $c^2 < n$, we achieve lower running time
and better scaling when every thread builds a temporary graph data structure. In this contraction step, the vertices are assigned to thread at runtime. These temporary graphs are then combined by a single thread.

\paragraph{Parallel Padberg-Rinaldi.}
In parallel, we run the Padberg-Rinaldi reductions $1$, $3$ and $4$ on the contracted graph. As these criteria are local and independent, they can be parallelized trivially. Updates to the union-find data structure are inside of a \texttt{critical} section. Performing reduction $2$ in parallel would entail in additional locks, as the vertex weights need to be updated on edge contraction.

\subsection{Further Details}
\label{sec:shuffle}
The label propagation algorithm by Raghavan
et al.~\cite{raghavan2007near} traverses the graph vertices in random
order. Other implementations of the algorithm~\cite{staudt2013engineering} omit
this explicit randomization and rely on implicit randomization through
parallelism, as the vertex processing order in parallel label propagation is
non-deterministic. Our implementation to find the new label of a vertex $v$ in uses an array, in which we sum up the weights for all
clusters in the neighborhood $N(v)$. Therefore randomizing the vertex traversal
order would destroy any graph locality, leading to many random reads in the
large array, which is very cache inefficient. Thus we trade off
randomness and graph locality by randomly shuffling small blocks of vertex ids
but traversing these shuffled blocks successively.

Using a time-forward processing technique~\cite{zeh2002efficient} the label propagation as well as the contraction algorithm can be implemented in external memory~\cite{akhremtsev2014semi} using Sort($|E|$) I/Os overall. 
Hence, if we only use the label propagation contraction technique in external memory and use the whole algorithm as soon as the graph fits into internal memory, we directly obtain an external memory algorithm for the minimum cut problem. We do not further investigate this variant of the algorithm as our focus is on fast internal memory algorithms for the problem. 

\section{Experiments}\label{s:experiments}

In this section we compare our algorithm \texttt{\algname} with existing algorithms for the
minimum cut problem on real-world and synthetic graphs.
We compare the sequential variant of our algorithm to
efficient implementations of existing algorithms and show how our algorithm
scales on a shared-memory machine.

\subsubsection{Experimental Setup and Methodology.}

We implemented the algorithms using \CC-17 and compiled all
codes using g++-7.1.0 with full optimization (\texttt{-O3}). All of our experiments are conducted on a
machine with two Intel Xeon E5-2643 v4 with 3.4GHz with 6 CPU cores each and
1.5 TB RAM in total. 
In general, we perform five repetitions per instance and report the average running time as well as cut.

\subsubsection{Algorithms.}  We compare our algorithm with our
implementations of the algorithm of Nagamochi, Ono and
Ibaraki~(\texttt{NOI})~\cite{nagamochi1994implementing} and the
$(2+\varepsilon)$-approximation algorithm of Matula
(\texttt{Matula})~\cite{matula1993linear}. In addition, we compare against the algorithm of Hao and Orlin~(\texttt{HO})~\cite{hao1992faster} by using the implementation of
Chekuri et al.~\cite{Chekuri:1997:ESM:314161.314315}. We also performed
experiments with Chekuri~\etal's implementations of \texttt{NOI}, but our implementation
is generally faster.
For \texttt{HO}, Chekuri~\etal give
variants with and without Padberg-Rinaldi tests and with an excess detection
heuristic~\cite{Chekuri:1997:ESM:314161.314315}, which contracts nodes with
large pre-flow excess. We use three variants of the algorithm of Hao and Orlin in our experiments: \texttt{HO\_A} uses Padberg-Rinaldi tests, \texttt{HO\_B}
uses excess detection and \texttt{HO\_C} uses both.
We also use their implementation of the algorithm of Karger and
Stein~\cite{karger1996new,code,Chekuri:1997:ESM:314161.314315} (\texttt{KS}) without Padberg-Rinaldi tests. We
only perform a single iteration of their algorithm, as this is already slower than
all other algorithms. Note that performing more iterations yields a smaller error probability, but also makes the algorithm even slower. The implementation crashes on very large instances due to overflows in the graph data structure used to for edge contractions.
We do not include the algorithm by Stoer and Wagner~\cite{stoer1997simple}, as it is far slower than \texttt{NOI} and
\texttt{HO}~\cite{Chekuri:1997:ESM:314161.314315,junger2000practical} and was
also slower in preliminary experiments we conducted. We also do not include the near-linear algorithm of Henzinger \etal~\cite{henzinger2017local}, as the other algorithms are quasi linear in most instances examined and the algorithm of Henzinger \etal has large constants. We performed preliminary experiments with the core of the algorithm, which indicate that the algorithm is slower in practice.

\subsubsection{Instances.}
We perform experiments on clustered Erd\H{o}s-Rényi graphs that are generated using the generator from Chekuri~\etal\cite{Chekuri:1997:ESM:314161.314315}, which are commonly used in the literature~\cite{nagamochi1994implementing,junger2000practical,Chekuri:1997:ESM:314161.314315,padberg1990efficient}. We also perform experiments on random hyperbolic graphs~\cite{krioukov2010hyperbolic, von2015generating} and on large undirected real-world graphs taken from the 10th DIMACS Implementation Challenge~\cite{benchmarksfornetworksanalysis} as well as the Laboratory for Web Algorithmics~\cite{BRSLLP,BoVWFI}. As these graphs contain vertices with low degree (and therefore trivial cuts), we use the $k$-core decomposition~\cite{batagelj2003m}, which gives the largest subgraph, in which each vertex has a degree of at least~$k$, to generate input graphs. We use the largest connected components of these core graphs to generate graphs in which the minimum cut is not trivial. For every real-world graph, we use $k$-cores for four different values of $k$. Appendix~\ref{app:instances} shows basic properties of all graph instances in more detail. The graphs used in our experiments have up to $70$ million vertices and up to $5$ billion edges. To the best of our knowledge, these graphs are the largest instances reported in literature to be used for experiments on global minimum cuts.

\subsubsection{Configuring the Algorithm.}
In preliminary experiments we tuned the parameters of our algorithm on smaller random hyperbolic graphs. These experiments have been performed on instances not used for the evaluation here. We detail these experiments in Appendix~\ref{app:config}, and omit them here due to space constraints. In further experiments conducted in this section, we use the configuration of \texttt{\algname} given by the parameter tuning in Appendix~\ref{app:config}, which performs two iterations of LPA and randomly shuffles blocks of $128$ vertices each. We set the bound $n_0$ to \numprint{10000} and did not encounter a single instance with more than a single bulk contraction step.

\subsection{Experimental Results}

\paragraph{Clustered Erd\H{o}s-Rényi Graphs.}




\begin{figure}[tb]
\resizebox{\textwidth}{!}{
  \begin{tikzpicture}[small dot/.style={fill=black,circle,scale=0.25}]
    \begin{axis}[
      name=plot1,
      xlabel={Number of Vertices},
      ylabel={Running Time per Edge $(ns)$},
      scaled ticks=false,
      xtick={12500,25000,50000,100000,150000,200000},
      xticklabels={12.5K,25K,50K,100K,150K,200K},
      xticklabel style={
        /pgf/number format/fixed,
      },
      title={$n=12.5K$ - $200K, d=10\%, k=2$},
    legend style={font=\normalsize},
    label style={font=\normalsize},
    title style={font=\normalsize},
    tick label style={font=\normalsize},
      x tick label style={rotate=60, anchor=east},
      x label style={at={(axis description cs:0.5,-0.06)},anchor=north},
      ymode=log,
      ]
      \addplot coordinates { (12500,181.2625)(25000,209.1285)(50000,212.036)};
      \addlegendentry{algo=ks\_ks\_nopr};
      \addplot coordinates { (12500,321.05) (25000,367.503) (50000,405.68) (100000,470.561) };
      \addlegendentry{algo=ks\_ho\_noxs};
      \addplot coordinates { (12500,193.936) (25000,223.625) (50000,241.813) (100000,257.257) };
      \addlegendentry{algo=ks\_ho\_nopr};
      \addplot coordinates { (12500,241.299) (25000,318.125) (50000,365.727) (100000,331.197) };
      \addlegendentry{algo=ks\_ho};
      \addplot coordinates { (12500,73.9655) (25000,89.2893) (50000,101.333) (100000,129.239) (150000,148.953) (200000,81.1435) };
      \addlegendentry{algo=noi};
      \addplot coordinates { (12500,38.2732) (25000,43.7983) (50000,46.3199) (100000,51.0396) (150000,60.3919) (200000,49.2056) };
      \addlegendentry{algo=matula};
      \addplot coordinates { (12500,23.002) (25000,24.3954) (50000,26.391) (100000,30.1682) (150000,29.6003) (200000,28.9474) };
      \addlegendentry{algo=viecut};

      \legend{};

\end{axis}

    \begin{axis}[
      name=plot3,
      at=(plot1.right of south east),
      title={$n=100K, d=10\%, k=2$ - $128$},
    xlabel={Number of Clusters},
    xshift=6mm,
    scaled ticks=false,
    xtick=data,
    xmode=normal,
    xticklabels={2,,8,16,32,64,128},
    xticklabel style={
      align=center,
        /pgf/number format/fixed,
      },
      legend pos=outer north east,
      ymode=log,
    legend style={font=\normalsize},
    label style={font=\normalsize},
    title style={font=\normalsize},
    tick label style={font=\normalsize},
    clip=false,
    ]

\addplot[draw=none] coordinates { (2,235.2805)(4,333.6755)(8,410.554)(16,499.456)(32,672.02)(64,704.015)(128,1012.255)};
\addlegendentry{\texttt{KS}};
\addplot coordinates { (2,470.561) (4,667.351) (8,821.108) (16,998.912) (32,1344.04) (64,1408.03) (128,2024.51) };
\addlegendentry{\texttt{HO\_A}};
\addplot coordinates { (2,257.257) (4,418.85) (8,546.843) (16,685.591) (32,658.761) (64,682.587) (128,893.877) };
\addlegendentry{\texttt{HO\_B}};
\addplot coordinates { (2,331.197) (4,519.089) (8,683.823) (16,849.38) (32,912.577) (64,988.97) (128,1128.5) };
\addlegendentry{\texttt{HO\_C}};
\addplot coordinates { (2,129.239) (4,186.736) (8,319.451) (16,197.922) (32,113.395) (64,61.9847) (128,51.4503) };
\addlegendentry{\texttt{NOI}};
\addplot coordinates { (2,51.0396) (4,46.2595) (8,44.3773) (16,41.673) (32,40.7127) (64,38.2494) (128,37.0853) };
\addlegendentry{\texttt{Matula}};
\addplot coordinates {(2,29.8777) (4,25.1355) (8,22.394) (16,21.2625) (32,20.907) (64,20.9411) (128,21.6541) };
\addlegendentry{\texttt{seq\algname}};
\legend{}


\end{axis}    \begin{axis}[
      name=plot2,
      at=(plot3.right of south east),
      xlabel={Density [\%]},
      xshift=6mm,
      scaled ticks=false,
      title={$n=100K, d=2.5\%$ - $100\%, k=2$},
      xtick={2.5,5,10,25,50,75,100},
      xticklabels={2.5,,10,25,50,75,100},
    xticklabel style={
      /pgf/number format/fixed,
      align=center
      },
      x label style={at={(axis description cs:0.5,-0.00)},anchor=north},
      ymode=log,
    legend style={font=\normalsize},
    label style={font=\normalsize},
    title style={font=\normalsize},
    tick label style={font=\normalsize},
    legend pos=outer north east,
    clip=false,
    ]

\addplot coordinates { (2.5,222.946)};
\addlegendentry{\texttt{KS}};
\addplot coordinates { (2.5,480.189) (5.0,438.62) (10.0,470.527) };
\addlegendentry{\texttt{HO\_A}};
\addplot coordinates { (2.5,243.994) (5.0,250.043) (10.0,258.021) };
\addlegendentry{\texttt{HO\_B}};
\addplot coordinates { (2.5,306.387) (5.0,353) (10.0,329.593) };
\addlegendentry{\texttt{HO\_C}};
\addplot coordinates { (2.5,109.577) (5.0,107.956) (10.0,107.76) (25.0,79.8169) (50.0,61.6961) (75.0,57.2839) (100.0,57.8905) };
\addlegendentry{\texttt{NOI}};
\addplot coordinates { (2.5,47.6882) (5.0,48.1773) (10.0,50.3732) (25.0,43.6194) (50.0,38.4107) (75.0,39.0852) (100.0,40.0415) };
\addlegendentry{\texttt{Matula}};
\addplot coordinates { (2.5,29.6736) (5.0,29.6869) (10.0,29.8777) (25.0,28.6927) (50.0,27.4626) (75.0,27.0509) (100.0,26.7951) };
\addlegendentry{\texttt{seq\algname}};

\end{axis}
\end{tikzpicture}
}
\caption{Total running time in nanoseconds per edge in clustered Erd\H{o}s-Rényi graphs\label{fig:test3}}
\end{figure}
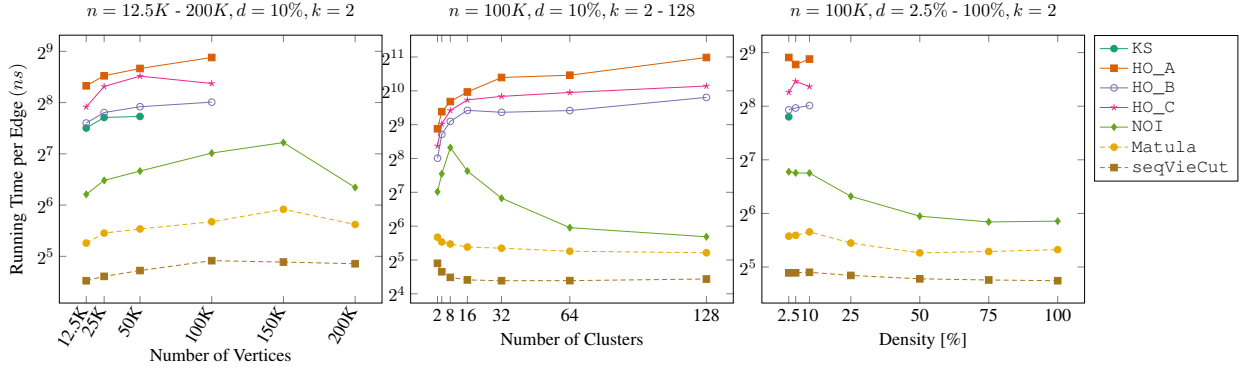

Clustered Erd\H{o}s-R\'enyi graphs have distinct small cuts between the clusters and do not have any other small cuts.
We perform three experiments varying one parameter of the graph class and use default parameters for the other two parameters.
Our default parameters are $n=\numprint{100000}$, $d=10\%$ and $k=2$. 
The code of Chekuri~\etal\cite{Chekuri:1997:ESM:314161.314315} uses 32 bit integers to store vertices and edges. We could therefore not perform the experiments with $m \geq 2^{31}$ with~\texttt{HO}.
Figure~\ref{fig:test3} shows the results for these experiments. 
First of all, on $20\%$ of the instances \texttt{KS} returns non-optimal results.
No other algorithm returned any non-optimal minimum cuts on any graph of this dataset. 
Moreover, \texttt{seq\algname} is the fastest algorithm on all of these instances, followed by \texttt{Matula}, which is $40$ to $100\%$ slower on these instances.

\texttt{seq\algname} is faster on graphs with a lower number of vertices, as the array containing cluster affiliations -- which has one entry per vertex and is accessed for each edge -- fits into cache. In graphs with $k=2,4,8$, the final number of clusters in the label propagation algorithm is equal to $k$, as label propagation correctly identifies the clusters. In the graph contraction step, we iterate over all edges and check whether the incident vertices are in different clusters. For this branch, the compiler assumes that they are indeed in different cluster. However, in these graphs, the chance for any two adjacent nodes being in the same cluster is $\frac{1}{k}$, which is far from zero. This results in a large amount of branch misses (for $n=$ \numprint{100000}, $d=10\%$, $k=2$: average $14\%$ branch misses, in total $1.5$ billion missed branches). Thus the performance is better with higher values of $k$. The fastest exact algorithm is \texttt{NOI}. This matches the experiments by Chekuri
\etal\cite{Chekuri:1997:ESM:314161.314315}.

\paragraph{Random Hyperbolic Graphs.}

We now perform experiments on random hyperbolic graphs with $n=2^{20} - 2^{25}$ and an average degree of $2^5 - 2^8$. We generated $3$ graphs for each of the $24$ possible combinations yielding a total of $72$ RHG graphs. Note that these graphs are hard instances for the inexact algorithms, as they contain few -- usually only one -- small cuts and both sides of the cut are large. From a total of $360$ runs, \texttt{seq\algname} does not return the correct minimum cut in $1\%$ of runs and \texttt{Matula} does not return the correct minimum cut in $31\%$ of runs. \texttt{KS}, which crashes on large instances, returns non-optimal cuts in $52\%$ of the runs where it ran to completion.
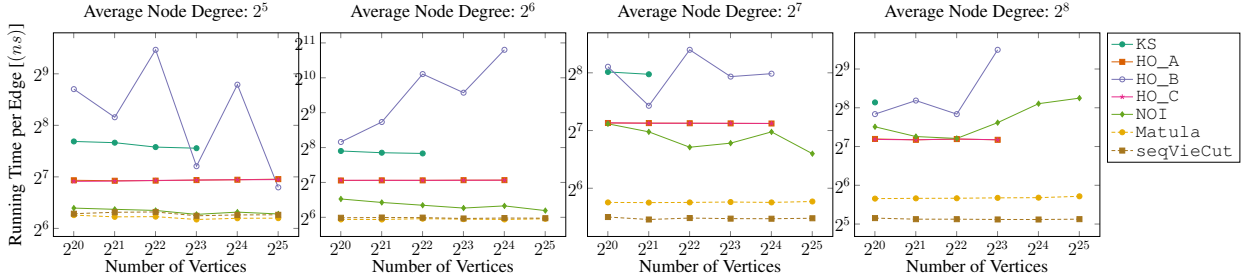
\begin{figure}[tb]
\resizebox{\textwidth}{!}{
  \begin{tikzpicture}
    \begin{axis}[
      name=plot1,
    title={Average Node Degree: $2^5$},
    xlabel={Number of Vertices},
    ylabel={Running Time per Edge [$(ns)$]},
    ymode=log,
    xmode=log,
    ]
\addplot coordinates { (1.04858e+06,205.952)(2.09715e+06,202.578)(4.1943e+06,190.9492)(8.38861e+06,188.0852)};
\addlegendentry{algo=ks\_ks\_nopr};
\addplot coordinates { (1.04858e+06,122.172)(2.09715e+06,121.4196)(4.1943e+06,121.636)(8.38861e+06,122.5016)(1.67772e+07,123.0904)(3.35544e+07,124.0584)};
\addlegendentry{algo=ks\_ho\_noxs};
\addplot coordinates { (1.04858e+06,416.276)(2.09715e+06,285.1668)(4.1943e+06,706.544)(8.38861e+06,147.6556)(1.67772e+07,442.028)(3.35544e+07,111.0448)};
\addlegendentry{algo=ks\_ho\_nopr};
\addplot coordinates { (1.04858e+06,120.4252)(2.09715e+06,120.8632)(4.1943e+06,121.954)(8.38861e+06,122.5212)(1.67772e+07,123.0476)(3.35544e+07,123.3104)};
\addlegendentry{algo=ks\_ho};
\addplot coordinates { (1.04858e+06,84.0544)(2.09715e+06,82.6796)(4.1943e+06,81.4332)(8.38861e+06,77.2488)(1.67772e+07,79.6376)(3.35544e+07,77.6924)};
\addlegendentry{algo=noi};
\addplot coordinates { (1.04858e+06,76.49)(2.09715e+06,74.7188)(4.1943e+06,74.9352)(8.38861e+06,72.108)(1.67772e+07,73.3104)(3.35544e+07,73.406)};
\addlegendentry{algo=matula};
\addplot coordinates { (1.04858e+06,77.968)(2.09715e+06,79.4508)(4.1943e+06,79.8168)(8.38861e+06,75.8128)(1.67772e+07,76.6912)(3.35544e+07,76.814)};
\addlegendentry{algo=viecut128};

\legend{}
\end{axis}
\begin{axis}[
  name=plot2,
  at=(plot1.right of south east),
    title={Average Node Degree: $2^6$},
    xlabel={Number of Vertices},
    xshift=6mm,
    ymode=log,
    xmode=log,
    ]

\addplot coordinates { (1.04858e+06,238.958)(2.09715e+06,230.87)(4.1943e+06,227.548)};
\addlegendentry{algo=ks\_ks\_nopr};
\addplot coordinates { (1.04858e+06,133.0016)(2.09715e+06,133.3508)(4.1943e+06,133.3048)(8.38861e+06,133.7876)(1.67772e+07,133.888)};
\addlegendentry{algo=ks\_ho\_noxs};
\addplot coordinates { (1.04858e+06,285.85)(2.09715e+06,424.952)(4.1943e+06,1101.568)(8.38861e+06,761.028)(1.67772e+07,1784.076)};
\addlegendentry{algo=ks\_ho\_nopr};
\addplot coordinates { (1.04858e+06,133.5584)(2.09715e+06,133.4304)(4.1943e+06,133.414)(8.38861e+06,133.6036)(1.67772e+07,133.8432)};
\addlegendentry{algo=ks\_ho};
\addplot coordinates { (1.04858e+06,92.0936)(2.09715e+06,85.9556)(4.1943e+06,81.3072)(8.38861e+06,76.916)(1.67772e+07,80.2744)(3.35544e+07,73.1332)};
\addlegendentry{algo=noi};
\addplot coordinates { (1.04858e+06,61.3152)(2.09715e+06,61.2596)(4.1943e+06,62.2768)(8.38861e+06,61.702)(1.67772e+07,61.324)(3.35544e+07,62.1496)};
\addlegendentry{algo=matula};
\addplot coordinates { (1.04858e+06,63.4084)(2.09715e+06,63.726)(4.1943e+06,63.604)(8.38861e+06,62.5364)(1.67772e+07,63.08)(3.35544e+07,62.8584)};
\addlegendentry{algo=viecut128};

\legend{}
\end{axis}

\begin{axis}[
  name=plot3,
  at=(plot2.right of south east),
  xshift=6mm,
  title={Average Node Degree: $2^7$},
  xlabel={Number of Vertices},
  ymode=log,
  xmode=log,
  ]
\addplot coordinates { (1.04858e+06,258.47)(2.09715e+06,251.366)};
\addlegendentry{algo=ks\_ks\_nopr};
\addplot coordinates { (1.04858e+06,140.0188)(2.09715e+06,139.852)(4.1943e+06,139.5908)(8.38861e+06,139.358)(1.67772e+07,139.1868)};
\addlegendentry{algo=ks\_ho\_noxs};
\addplot coordinates { (1.04858e+06,275.0988)(2.09715e+06,172.1892)(4.1943e+06,337.2444)(8.38861e+06,244.4)(1.67772e+07,253.1792)};
\addlegendentry{algo=ks\_ho\_nopr};
\addplot coordinates { (1.04858e+06,140.3176)(2.09715e+06,139.5636)(4.1943e+06,139.6456)(8.38861e+06,139.512)(1.67772e+07,139.2912)};
\addlegendentry{algo=ks\_ho};
\addplot coordinates { (1.04858e+06,138.6472)(2.09715e+06,125.9004)(4.1943e+06,104.672)(8.38861e+06,109.9256)(1.67772e+07,125.8436)(3.35544e+07,96.8648)};
\addlegendentry{algo=noi};
\addplot coordinates { (1.04858e+06,53.9548)(2.09715e+06,53.8336)(4.1943e+06,53.9616)(8.38861e+06,54.2112)(1.67772e+07,53.9612)(3.35544e+07,54.5576)};
\addlegendentry{algo=matula};
\addplot coordinates { (1.04858e+06,45.276)(2.09715e+06,43.9496)(4.1943e+06,44.7232)(8.38861e+06,44.4136)(1.67772e+07,44.3692)(3.35544e+07,44.604)};
\addlegendentry{algo=viecut128};

\legend{}
\end{axis}

\begin{axis}[
  name=plot4,
  at=(plot3.right of south east),
  xshift=6mm,
  title={Average Node Degree: $2^8$},
  xlabel={Number of Vertices},
  legend pos=outer north east,
  ymode=log,
  xmode=log,
  ]
\addplot coordinates { (1.04858e+06,281.546)};
\addlegendentry{\texttt{KS}};
\addplot coordinates { (1.04858e+06,146.462)(2.09715e+06,144.2432)(4.1943e+06,146.1424)(8.38861e+06,144.2832)};
\addlegendentry{\texttt{HO\_A}};
\addplot coordinates { (1.04858e+06,228.3808)(2.09715e+06,290.9676)(4.1943e+06,228.3404)(8.38861e+06,721.352)};
\addlegendentry{\texttt{HO\_B}};
\addplot coordinates { (1.04858e+06,145.2144)(2.09715e+06,144.2732)(4.1943e+06,146.3564)(8.38861e+06,144.3116)};
\addlegendentry{\texttt{HO\_C}};
\addplot coordinates { (1.04858e+06,182.0368)(2.09715e+06,152.9244)(4.1943e+06,147.824)(8.38861e+06,196.0372)(1.67772e+07,275.1888)(3.35544e+07,303.9144)};
\addlegendentry{\texttt{NOI}};
\addplot coordinates { (1.04858e+06,50.4004)(2.09715e+06,50.64)(4.1943e+06,50.7244)(8.38861e+06,51.0832)(1.67772e+07,51.2696)(3.35544e+07,52.5148)};
\addlegendentry{\texttt{Matula}};
\addplot coordinates { (1.04858e+06,35.568)(2.09715e+06,34.892)(4.1943e+06,34.88836)(8.38861e+06,34.71568)(1.67772e+07,34.69084)(3.35544e+07,34.8996)};
\addlegendentry{\texttt{seq\algname}};

\end{axis}
\end{tikzpicture}
}
\caption{Total running time in nanoseconds per edge in RHG graphs \label{fig:rhgtests}}
\end{figure}
Figure~\ref{fig:rhgtests} shows the results for these experiments. On nearly all of these graphs, \texttt{NOI} is faster than \texttt{HO}. On the sparse graphs with an average degree of $2^5$, \texttt{seq\algname}, \texttt{Matula} and \texttt{NOI} nearly have equal running time. On denser graphs with an average degree of $2^8$, \texttt{seq\algname} is $40\%$ faster than \texttt{Matula} and $4$ to $10$ times faster than \texttt{NOI}. \texttt{HO\_A} and \texttt{HO\_C} use preprocessing with the Padberg-Rinaldi heuristics.
Multiple iterations of this preprocessing contract the RHG graph into two nodes.
The running time of those algorithms is $50\%$ higher on sparse graphs and $4$ times higher on dense graphs compared to \texttt{seq\algname}.
Figure~\ref{fig:breakdown} shows a time breakdown for \texttt{seq\algname} on large RHG graphs with $n=2^{25}$. Around $80\%$ of the running time is in the label propagation step and the rest is mostly spent in graph contraction. The correcting step has low running time on most graphs, as it is not performed on large clusters.


\paragraph{Real-World Graphs.}

We now run experiments on six large real-world social and web graphs. Most of these graphs have many small cuts. On these graphs, there are no non-optimal minimum cuts for any algorithm except for \texttt{KS}, which has $36\%$ non-optimal results. Figure~\ref{fig:real} gives slowdown plots to the fastest algorithm (\texttt{seq\algname} in each case) for the real-world graphs. On these graphs, \texttt{seq\algname} is the fastest algorithm, far faster than the other algorithms. \texttt{Matula} is not much faster than \texttt{NOI}, as most of the running time is in the first iteration of their BFS algorithm, which is similar for both algorithms. On the largest real-world graphs, \texttt{seq\algname} is approximately $3$ times faster than the next fastest algorithm \texttt{Matula}. We also see that \texttt{seq\algname}, \texttt{Matula} and \texttt{NOI} all perform better on denser graphs. For \texttt{Matula} and \texttt{NOI}, this can most likely be explained by the smaller vertex priority queue. For \texttt{seq\algname}, this is mainly due to better cache locality. As \texttt{HO} does not benefit from denser graphs, it has high slow down on dense graphs.

The highest speedup in our experiments is in the $10$-core of \texttt{gsh-2015-host}, where \texttt{seq\algname} is faster than the next fastest algorithm (\texttt{Matula}) by a factor of $4.85$. The lowest speedup is in the $25$-core of \texttt{twitter-2010}, where \texttt{seq\algname} is $50\%$ faster than the next fastest algorithm (\texttt{HO\_B}). The average speedup factor of \texttt{seq\algname} to the next fastest algorithm is $2.37$.
\texttt{NOI} and \texttt{Matula} perform badly on the cores of the graph \texttt{twitter-2010}. This graph has a very low diameter (average distance is $4.46$), and as a consequence the priority queue used in these algorithms is filled far quicker than in graphs with higher diameter. Therefore the priority queue operations become slow and the total running time is very high.

\begin{figure}[bt]
\centering
\includegraphics[width=8cm]{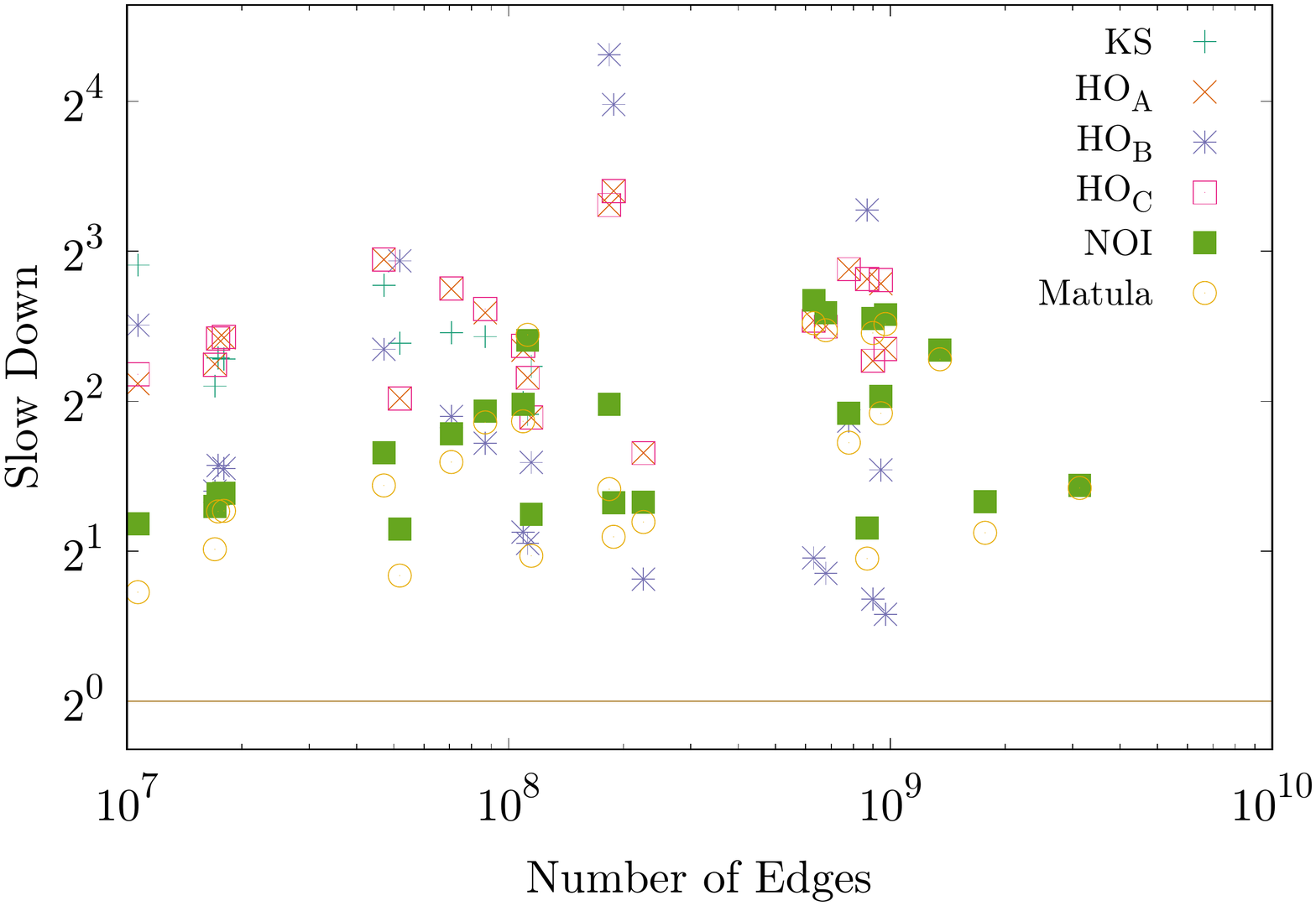}
\includegraphics[width=8cm]{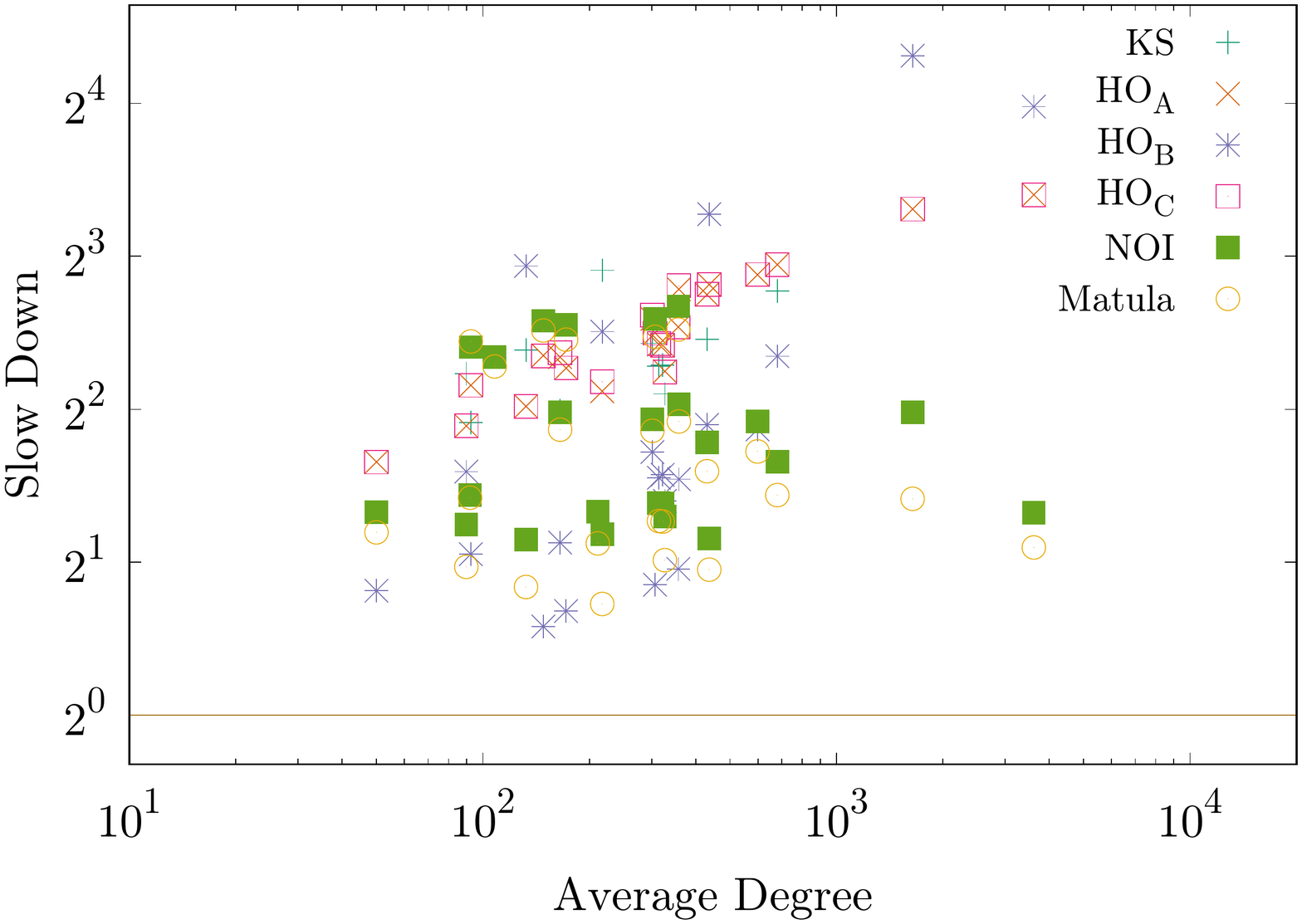}

\caption{Slowdowns of competitors to \texttt{\algname} in large real-world graphs. We display slowdowns based on the absolute number of edges (left), and by the average vertex degree (right) in the graph\label{fig:real}}
\end{figure}

To summarize, both in generated and real-world graphs, even in sequential runs \texttt{seq\algname} is up to a factor of $6$ faster than the state of the art, while achieving a high solution quality even for hard instances such as the hyperbolic graphs. The performance of \texttt{seq\algname} is especially good on the real-world graphs, presumably as these graphs have high locality.

\paragraph{Shared-Memory Parallelism.}


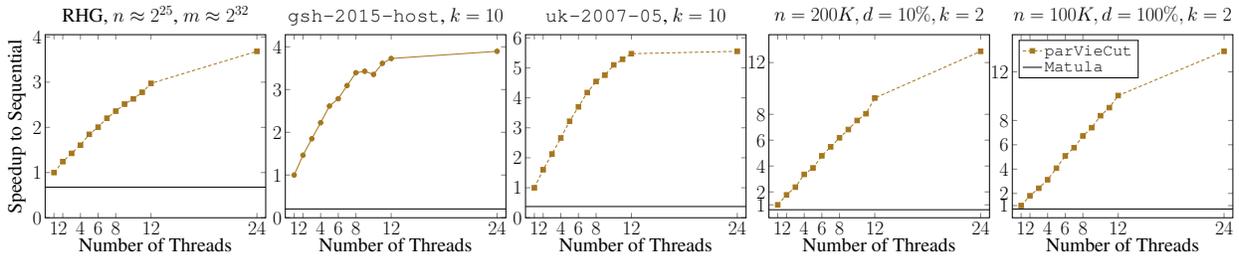
\begin{figure}[t]
  \centering
\resizebox{\textwidth}{!}{
  \definecolor{tempc}{RGB}{166,118,29}
  \definecolor{tempc2}{RGB}{102,102,102}
  \begin{tikzpicture}
    \begin{axis}[
      name=plot1,
    title={RHG, $n \approx 2^{25}$, $m \approx 2^{32}$},
    xlabel={Number of Threads},
    ylabel={Speedup to Sequential},
    xmin=0,xmax=25,
    ymin=0,
    ytick={0,1,2,3,4,5,6},
    legend style={font=\LARGE},
    label style={font=\LARGE},
    title style={font=\LARGE},
    tick label style={font=\LARGE},
    xtick={1,2,4,6,8,12,24},
    ]

\addplot +[color=tempc,style=solid,mark=square*,densely dashed] coordinates { (1,1) (2,1.24103) (3,1.42498) (4,1.60555) (5,1.84591) (6,2.00426) (7,2.2023) (8,2.35885) (9,2.51431) (10,2.62986) (11,2.77616) (12,2.9733) (24,3.6823) };
\addlegendentry{algo=parviecut};

    \addplot[mark=none, black] coordinates {(0,74.86/110.55) (25,74.86/110.55)};
  \addlegendentry{\texttt{Matula}}
\legend{}
\end{axis}

\begin{axis}[
  name=plot2,
  at=(plot1.right of south east),
    title={\texttt{gsh-2015-host}, $k=10$},
    xlabel={Number of Threads},
    xshift=5mm,
    xmin=0,xmax=25,
    ymin=0,
    ytick={0,1,2,3,4,5,6},
    legend style={font=\LARGE},
    label style={font=\LARGE},
    title style={font=\LARGE},
    tick label style={font=\LARGE},
    xtick={1,2,4,6,8,12,24},
    ]

    \addplot +[color=tempc] coordinates { (1,1) (2,1.46386) (3,1.84777) (4,2.22353) (5,2.61419) (6,2.7847) (7,3.09315) (8,3.39319) (9,3.42698) (10,3.35361) (11,3.61256) (12,3.72593) (24,3.89444) };
    \addlegendentry{algo=parviecut};

    \addplot[mark=none, black] coordinates {(0,48.71/236.57) (25,48.71/236.57)};
  \addlegendentry{\texttt{Matula}}

\legend{}
\end{axis}

\begin{axis}[
  name=plot3,
  at=(plot2.right of south east),
  xshift=5mm,
    title={\texttt{uk-2007-05}, $k=10$},
    xlabel={Number of Threads},
    xmin=0,xmax=25,
    ymin=0,
    ytick={0,1,2,3,4,5,6},
    legend style={font=\LARGE},
    label style={font=\LARGE},
    title style={font=\LARGE},
    tick label style={font=\LARGE},
    xtick={1,2,4,6,8,12,24},
    ]

   \addplot +[color=tempc,style=solid,mark=square*,densely dashed] coordinates { (1,1) (2,1.60205) (3,2.12522) (4,2.66005) (5,3.21791) (6,3.69778) (7,4.17485) (8,4.54357) (9,4.75413) (10,5.09668) (11,5.28345) (12,5.4705) (24,5.54966) };
   \addlegendentry{algo=parviecut};

  \addplot[mark=none, black] coordinates {(0,67/179) (25,67/179)};
  \addlegendentry{\texttt{Matula}}

\legend{}
\end{axis}

 \begin{axis}[
   name=plot4,
   at=(plot3.right of south east),
   xshift=6mm,
   title={$n=200K, d=10\%, k=2$},
   xlabel={Number of Threads},
    xmin=0,xmax=25,
    ymin=0,
    ytick={1,2,4,6,8,12},
    legend style={font=\LARGE},
    label style={font=\LARGE},
    title style={font=\LARGE},
    tick label style={font=\LARGE},
    xtick={1,2,4,6,8,12,24},
   ]

    \addplot +[color=tempc,style=solid,mark=square*,densely dashed] coordinates { (1,1.00005) (2,1.76791) (3,2.36689) (4,3.35236) (5,3.8532) (6,4.79107) (7,5.48427) (8,6.18399) (9,6.8231) (10,7.51987) (11,8.04387) (12,9.25958) (24,12.8705) };
    \addlegendentry{algo=parvc\_nc};

    \addplot[mark=none, black] coordinates {(0,122.2/196.82) (25,122.2/196.82)};
  \addlegendentry{\texttt{Matula}}
\legend{};
\end{axis}

\begin{axis}[
  name=plot5,
  at=(plot4.right of south east),
    title={$n=100K, d=100\%, k=2$},
    xlabel={Number of Threads},
    xshift=6mm,
    legend pos=north west,
    xmin=0,xmax=25,
    ymin=0,
    ytick={1,2,4,6,8,12},
    label style={font=\LARGE},
    title style={font=\LARGE},
    tick label style={font=\LARGE},
    xtick={1,2,4,6,8,12,24},
    ]

    \addplot +[color=tempc,style=solid,mark=square*,densely dashed] coordinates { (1,1.00022) (2,1.80779) (3,2.42722) (4,3.12729) (5,4.07675) (6,5.08563) (7,5.7606) (8,6.72386) (9,7.40532) (10,8.38926) (11,9.05278) (12,10.0465) (24,13.6835) };
    \addlegendentry{\texttt{par\algname}};

    \addplot[mark=none, black] coordinates {(0,288/400.41) (25,288/400.41)};
  \addlegendentry{\texttt{Matula}}

\end{axis}
\end{tikzpicture}
}
\caption{Speedup on large graphs over our algorithm using 1 thread.\label{scaling}}
\end{figure}
Figure~\ref{scaling} shows the speedup of \texttt{par\algname} compared to the sequential variant and to the next fastest algorithm, which is \texttt{Matula} in all of the large graph examined. 
We examine the largest graphs from each of the three graph classes and perform parallel runs using $1,2,3, \ldots, 12$ threads. We also perform experiments with $24$ threads, as the machine has $12$ cores and supports multi-threading.
On average, \texttt{par\algname} with $12$ is $6.3$ times faster than \texttt{seq\algname} ($24$ threads: $7.9$x faster), while still having a low error rate even on hard instances. On large RHG graphs we have $2$ non-optimal results out of $195$ runs. All other results of \texttt{par\algname} on large graphs were optimal.
Compared to the next fastest sequential algorithm \texttt{Matula}, this is an average speedup factor of $13.2$ ($24$ threads: $15.8$x faster). \texttt{par\algname} scales better on the graph \texttt{uk-2007-05} ($k=10$) and especially on the clustered Erd\H{o}s-Rényi graphs, presumably as these graphs contain many vertices with high degree.
Figure~\ref{fig:scaledown} shows average running time breakdowns averaged over all graphs. For this figure, the correcting algorithm is turned off for the two Erd\H{o}s-Rényi graphs. We can see that label propagation scales better than graph contraction. With one thread, label propagation uses $81\%$ of the total running time, with $24$ threads, it uses $71\%$ of the total running time.

\section{Conclusions and Future Work}\label{s:conclusions}

We presented the linear-time heuristic algorithm \texttt{\algname} for the minimum cut problem. \texttt{\algname} is based on the label propagation algorithm~\cite{raghavan2007near} and the Padberg-Rinaldi heuristics~\cite{padberg1990efficient}. Both on real-world graphs and a varied family of generated graphs, \texttt{\algname} is significantly faster than the state of art. Our algorithm has far higher solution quality than other heuristic algorithms while also being faster.
Important future work includes checking whether using better clustering techniques affect the observed error probability.
However, these clustering algorithms generally have higher running time.

\vfill
\pagebreak
\section*{Acknowledgements}
The research leading to these results has received funding from the European Research Council under the European Community's Seventh Framework Programme (FP7/2007-2013) /ERC grant agreement No. 340506

\renewcommand{\bibname}{\begin{flushleft} References \end{flushleft}}
\bibliographystyle{abbrv} 
\bibliography{quellen}

\begin{appendix}
  \section{Algorithm Overview}
  \label{app:algorithm}
\begin{algorithm2e}[h!]
\DontPrintSemicolon
\SetCommentSty{}
\caption{\texttt{\algname}}
\label{alg:generalalgorithm}
\textbf{input} graph $G=(V,E,c: V \to \MdN_{>{0}}), n_0 :$ bound for exact algorithm\;\;
$\mathcal{G} \leftarrow G$\;
\While{$|V_{\mathcal{G}}| > n_0$}{
\tcp*[l]{compute inexact kernel}
$\mathcal{C} \leftarrow$ computeClustering($\mathcal{G}$) \tcp*[r]{label propagation clustering}
$\mathcal{C}' \leftarrow$ fixMisplacedVertices($\mathcal{G}$, $\mathcal{C}$) \tcp*[r]{fix misplaced vertices}
$G_{\mathcal{C}'} \leftarrow$ contractClustering($\mathcal{G}$, $\mathcal{C}'$) \tcp*[r]{perform contraction} \;

\tcp*[l]{further apply exact reductions}
$\mathcal{E} \leftarrow$ findContractableEdges($G_{\mathcal{C}'}$) \tcp*[r]{find contractable edges with Padberg-Rinaldi} 
$\mathcal{G} \leftarrow$ contractEdges($G_{\mathcal{C}'}$, $\mathcal{E}$) \tcp*[r]{perform contraction}
}
$(A,B) \leftarrow $ NagamochiOnoIbaraki($\mathcal{G}$)   \tcp*[r]{solve minimum cut problem on final kernel} 
$(A',B') \leftarrow $ solutionTransfer($A,B$)   \tcp*[r]{transfer solution to input network}  \;

\textbf{return} $(A',B')$
\end{algorithm2e}
\section{Proofs}
\label{proof:lemma}
\setcounter{lemma}{0}
\begin{lemma}If there exist a minimum cut of $G$ such that each cluster of the clustering $\mathcal{C}$ is completely contained in one side of the minimum
  cut of $G$ and $|V_\mathcal{C}| > 1$, then $\lambda(G) = \lambda(G_\mathcal{C})$. \label{lemma:cluster}
\end{lemma}

\begin{proof} As node contraction removes cuts but does not add any new cuts,
$\lambda(G_\mathcal{C}) \geq \lambda(G)$ for each contraction with $|V_\mathcal{C}| > 1$. For an edge
$e$ in $G$, which is not part of some minimum cut of $G$, $\lambda(G) =
\lambda(G/e)$~\cite{karger1996new}. Contraction of a cluster $C$ in $G$ can also
be represented as the contraction of all edges in any spanning tree of $C$. If
the cluster $C$ is on one side of the minimum cut, none of the spanning edges
are part of the minimum cut. Thus we can contract each of the edges without
affecting the minimum cut of $G$. We can perform this contraction process on
each of the clusters and $\lambda(G_\mathcal{C}) = \lambda(G)$.
\end{proof}

\setcounter{lemma}{2}
\begin{lemma}
  The algorithm \texttt{\algname} has a running time complexity of $\Oh{n+m}$.
\end{lemma}

\begin{proof}
  One round of all reduction and contraction steps (Algorithm~\ref{alg:generalalgorithm}, lines 4-11) can be performed in $\Oh{n+m}$. The label propagation step contracts the graph by at least a factor of $2$, which yields geometrically shrinking graph size and thus a total running time of $\Oh{n+m}$.  We break this loop when the contracted graph has less than some constant $n_0$ number of vertices. The exact minimum cut of this graph with constant size can therefore be found in constant time. The solution transfer can be performed in linear time by performing the coarsening in reverse and pushing the two cut sides from each graph to the next finer~graph.

When the graph is not connected, throughout the algorithm one of the contracted graphs can contain isolated vertices, which our algorithm does not contract. However, when we discover an isolated vertex, our algorithm terminates, as the graph certainly has a minimum cut of $0$.
\end{proof}

  \section{Instances}
  \label{app:instances}

\subsubsection{Clustered Erd\H{o}s-Rényi Graphs.}

Many
previous experimental studies of minimum cut algorithms used a family of
clustered Erd\H{o}s-Rényi graphs with $m =
O(n^2)$~\cite{nagamochi1994implementing,junger2000practical,Chekuri:1997:ESM:314161.314315,padberg1990efficient}. This
family of graphs is specified by the following parameters:
number of vertices $n = |V|$, $d$ the graph density as a percentage where $m = |E| = \frac{n\cdot(n-1)}{2} \cdot \frac{d}{100}$ and the number of clusters $k$.
For each edge $(u, v)$, the integral edge weight $c(u,v)$ is generated
independently and uniformly in the interval $[1,100]$. When the vertices $u$ and
$v$ are in the same cluster, the edge weight is multiplied by $n$, resulting in
edge weights in the interval $[n, 100n]$. Therefore the minimum cut can be found
between two clusters with high probability. We performed three experiments on
this family of graphs. In each of these experiments we varied one of the graph
parameters and fixed the other two parameters. These experiments are similar to
older experiments~\cite{nagamochi1994implementing,junger2000practical,Chekuri:1997:ESM:314161.314315,padberg1990efficient} but scaled to larger
graphs to account for improvements in machine hardware.
We use the generator \emph{noigen} of Andrew Goldberg~\cite{code} to generate
the clustered Erd\H{o}s-Rényi graphs for these experiments. This generator was also
used in the study conducted by Chekuri \etal\cite{Chekuri:1997:ESM:314161.314315}. As our code uses the
METIS~\cite{karypis1998metis} graph format, we use a script to translate the
graph format. All experiments exclude I/O times.

\subsubsection{Random Hyperbolic Graphs (RHG) \cite{krioukov2010hyperbolic}.}

\label{rhggraphs}

Random hyperbolic graphs replicate many features of
real-world networks~\cite{chakrabarti2006graph}: the degree distribution follows
a power law, they often exhibit a community structure and have a small
diameter. In denser hyperbolic graphs, the minimum cut is often equal to the
minimum degree, which results in a trivial minimum cut. In order to prevent
trivial minimum cuts, we use a power law exponent of $5$. We use the generator
of von Looz \etal\cite{von2015generating}, which is a part of NetworKIT~\cite{staudt2014networkit}, to generate unweighted random hyperbolic
graphs with $2^{20}$ to $2^{25}$ vertices and an average vertex degree of $2^5$
to $2^8$. These graphs generally have very few small cuts and the minimum
cut has two partitions with similar sizes.

\begin{table}[tb]
  \setlength\intextsep{0pt}
  \centering
  \begin{tabular}{l|r|r||r|r|r|r|r}
    graph & $n$ & $m$ & $k$ & $n$ & $m$ & $\lambda$ & $\delta$ \\\hline\hline
    \texttt{hollywood-2011} & 2.2M & 114M & 20 & 1.3M & 109M & 1 & 20 \\
     \cite{BRSLLP,BoVWFI}    &&& 60 & 576K & 87M & 6 & 60\\
         &&& 100 & 328K & 71M & 77 & 100\\
         &&& 200 & 139K & 47M & 27 & 200\\\hline
    \texttt{com-orkut} & 3.1M & 117M & 16 & 2.4M & 112M & 14 & 16 \\
    \cite{BRSLLP,BoVWFI}    &&& 95 & 114K & 18M & 89 & 95\\
         &&& 98 & 107K & 17M & 76 & 98\\
         &&& 100 & 103K & 17M & 70 & 100\\\hline
    \texttt{uk-2002} & 18M & 262M & 10 & 9M & 226M & 1 & 10\\
    \cite{benchmarksfornetworksanalysis,BRSLLP,BoVWFI}     &&& 30 & 2.5M & 115M & 1 & 30\\
         &&& 50 & 783K & 51M & 1 & 50\\
         &&& 100 & 98K & 11M & 1 & 100\\\hline
    \texttt{twitter-2010} & 42M & 1.2B & 25 & 13M & 958M & 1 & 25 \\
    \cite{BRSLLP,BoVWFI}     &&& 30 & 10M & 884M & 1 & 30\\
         &&& 50 & 4.3M & 672M & 3 & 50\\
         &&& 60 & 3.5M & 625M & 3 & 60\\\hline
    \texttt{gsh-2015-host} & 69M & 1.8B & 10 & 25M & 1.3B & 1 & 10 \\
    \cite{BRSLLP,BoVWFI}     &&& 50 & 5.3M & 944M & 1 & 50\\
         &&& 100 & 2.6M & 778M & 1 & 100\\
         &&& 1000 & 104K & 188M & 1 & 1000\\\hline
    \texttt{uk-2007-05} & 106M & 3.3B & 10 & 68M & 3.1B & 1 & 10\\
    \cite{benchmarksfornetworksanalysis,BRSLLP,BoVWFI}     &&& 50 & 16M & 1.7B & 1 & 50 \\
         &&& 100 & 3.9M & 862M & 1 & 100\\
         &&& 1000 & 222K & 183M & 1 & 1000\\\hline
  \end{tabular}
  \caption{Statistics of real-world web graphs used in experiments. Original graph size and $k$-cores used in experiments with their respective minimum cuts\label{table:realworld}}
\end{table}

\subsubsection{Real-world Graphs.}
\label{rwgraphs}

We use large real-world web graphs and social
networks from~\cite{benchmarksfornetworksanalysis,BRSLLP,BoVWFI}, detailed in Table~\ref{table:realworld}. The minimum cut
problem on these web and social graphs can be seen as a network reliability
problem. As these graphs are generally disconnected and contain vertices with
very low degree, we use a $k$-core
decomposition~\cite{seidman1983network,batagelj2003m} to generate
versions of the graphs with a minimum degree of $k$. The $k$-core of a graph $G =
(V, E)$ is the maximum subgraph $G' = (V',E')$ with $V' \subseteq V$ and $E'
\subseteq E$, which fulfills the condition that every vertex in $G'$ has a
degree of at least $k$. We perform our experiments on the largest connected
component of $G'$. For every real-world graph we use, we compute a set of $4$
different $k$-cores, in which the minimum cut is not equal to the minimum degree
$k$.

We generate a diverse set of graphs with different sizes. On the large
graphs \texttt{gsh-2015-host} and \texttt{uk-2007-05}, we use cores with $k$ in 10, 50, 100
and 1000. In the smaller graphs we use cores with $k$ in 10, 30, 50 and
100. \texttt{twitter-2010} and \texttt{com-orkut} only had few cores where the
minimum cut is not trivial. Therefore we used those cores. As
\texttt{hollywood-2011} is very dense, we multiplied the $k$ value of all cores by a factor of 2.

        \vfill\pagebreak
\section{Algorithm Configuration}
\label{app:config}

We performed experiments to tune the number of label propagation iterations
and to find an appropriate amount of randomness for our algorithm. We conducted these experiments with different configurations on generated hyperbolic
graphs~(see Section~\ref{rhggraphs}) with $2^{15}$ to $2^{19}$ vertices with an average degree of $2^5$ to $2^8$ and compared error rate and running
time. The instances used here are different to the one used in the main text.

\begin{table}[b!]
  \centering
  \begin{tabular}{r||c|c|c|c|c|c|c}
&\texttt{\algname1} & \texttt{\algname2} &  \texttt{\algname3} & \texttt{\algname4} & \texttt{\algname5} & \texttt{\algname10} & \texttt{\algname25} \\\hline\hline
    \# of non optimal cuts & 29 & 14 & 15 & 16 & 19 & 19 & 18 \\ \hline
    average distance to opt. & 16.2\% & 2.44\% & 2.46\% & 3.30\% & 3.80\% & 3.37\% & 3.14\%\\
   \end{tabular}
   \caption{Error rate for configurations of \texttt{\algname} in RHG graphs (out of 300 instances)  \label{fig:rhg_it}}
\end{table}

Table~\ref{fig:rhg_it} shows the number of non-optimal cuts returned by
\texttt{\algname} with different numbers of label propagation iterations indicated by the integer in the name. Each
implementation traverses the graph in blocks of $256$ randomly shuffled elements
as described in Section~\ref{sec:shuffle}. The variant \texttt{\algname25}
performs up to 25 iterations or until the label propagation converges so that
only up to $\frac{1}{10000}$ of all nodes change their cluster. On average the
variant performed $20.4$ iterations. The results for all variants with $2$ to $25$ iterations are very similar with $14$ to $19$ non-optimal results and $2.44\%$ and $3.80\%$ average distance to the optimum. As the largest part of the total
running time is in the label propagation step, running the algorithm with a
lower amount of iterations is obviously faster. Therefore we use $2$ iterations
of label propagation in all of our experiments.

To compare the effect of graph traversal strategies, we compared different
configurations of our algorithm.  Configuration \texttt{\algname\_cons} does not
randomize the traversal order, \ie it traverses vertices consecutively by their ID, \texttt{\algname\_global} performs global
shuffling, \texttt{\algname\_fast} swaps each vertex with a random vertex with a
index distance up to 20. The configurations \texttt{\algname128},
\texttt{\algname256}, \texttt{\algname512}, \texttt{\algname1024} randomly shuffle
blocks of $128$, $256$, $512$, or $1024$ vertices and introduce randomness
without losing too much data locality. We also include the configurations
\texttt{par\algname\_cons} and \texttt{par\algname128}, which are shared-memory parallel implementation
with 12 threads. As a comparison, we also include the approximation algorithm of Matula
and a single run of the randomized algorithm of Karger and Stein.

\label{exp:shuffle}

\begin{figure}[tb]
\resizebox{\textwidth}{!}{
  \begin{tikzpicture}
    \begin{axis}[
      name=plot1,
    title={Average Node Degree: $2^5$},
    xlabel={Number of Vertices},
    ylabel={Running Time per Edge [$(ns)$]},
    ymode=log,
    xmode=log,
    ]
\addplot coordinates { (524288.0,130.898) (262144.0,138.72) (131072.0,138.895) (65536.0,145.577) (32768.0,145.934) };
\addlegendentry{algo=ks\_ks\_nopr};
\addplot coordinates { (524288.0,21.6873) (262144.0,18.4353) (131072.0,18.2545) (65536.0,17.8085) (32768.0,17.9732) };
\addlegendentry{algo=viecutglobal};
\addplot coordinates { (524288.0,14.6031) (262144.0,13.5002) (131072.0,14.6631) (65536.0,13.8989) (32768.0,14.9265) };
\addlegendentry{algo=matula};
\addplot coordinates { (524288.0,14.3158) (262144.0,12.6601) (131072.0,13.4684) (65536.0,13.1992) (32768.0,13.537) };
\addlegendentry{algo=viecut1024};
\addplot coordinates { (524288.0,14.327) (262144.0,12.9663) (131072.0,13.4842) (65536.0,13.895) (32768.0,13.3077) };
\addlegendentry{algo=viecut512};
\addplot coordinates { (524288.0,12.5559) (262144.0,10.7741) (131072.0,11.7454) (65536.0,12.1057) (32768.0,11.4421) };
\addlegendentry{algo=viecut256};
\addplot coordinates { (524288.0,11.9819) (262144.0,10.2563) (131072.0,11.0969) (65536.0,10.5812) (32768.0,10.7538) };
\addlegendentry{algo=viecut128};
\addplot coordinates { (524288.0,9.89499) (262144.0,8.63198) (131072.0,9.02028) (65536.0,9.11742) (32768.0,8.94738) };
\addlegendentry{algo=viecutfast};
\addplot coordinates { (524288.0,7.30613) (262144.0,6.86918) (131072.0,7.09254) (65536.0,7.51257) (32768.0,7.4812) };
\addlegendentry{algo=viecut};
\addplot coordinates { (524288.0,4.24657) (262144.0,3.67153) (131072.0,3.93438) (65536.0,4.20803) (32768.0,4.69048) };
\addlegendentry{algo=parviecut128};
\addplot coordinates { (524288.0,3.47864) (262144.0,3.20402) (131072.0,3.59721) (65536.0,3.29568) (32768.0,4.09491) };
\addlegendentry{algo=parviecut};
\legend{}
\end{axis}

\begin{axis}[
  name=plot2,
  at=(plot1.right of south east),
    title={Average Node Degree: $2^6$},
    xlabel={Number of Vertices},
    ymode=log,
    xshift=6mm,
    log basis y={2},
    xmode=log,
    log basis x={2},
    ]
\addplot coordinates { (32768.0,128.649) (65536.0,138.435) (131072.0,127.933) (262144.0,122.918) (524288.0,123.677) };
\addlegendentry{algo=ks\_ks\_nopr};
\addplot coordinates { (32768.0,21.6891) (65536.0,19.7414) (131072.0,25.2528) (262144.0,23.7565) (524288.0,23.4454) };
\addlegendentry{algo=viecutglobal};
\addplot coordinates { (32768.0,21.114) (65536.0,14.5729) (131072.0,15.8205) (262144.0,14.8901) (524288.0,14.5968) };
\addlegendentry{algo=matula};
\addplot coordinates { (32768.0,17.2824) (65536.0,14.5389) (131072.0,17.5047) (262144.0,16.3542) (524288.0,15.6368) };
\addlegendentry{algo=viecut1024};
\addplot coordinates { (32768.0,17.3416) (65536.0,14.8312) (131072.0,17.4834) (262144.0,16.3213) (524288.0,15.6466) };
\addlegendentry{algo=viecut512};
\addplot coordinates { (32768.0,18.8692) (65536.0,13.6537) (131072.0,17.7446) (262144.0,15.8307) (524288.0,15.2447) };
\addlegendentry{algo=viecut256};
\addplot coordinates { (32768.0,19.4317) (65536.0,11.6667) (131072.0,16.9475) (262144.0,14.7536) (524288.0,13.2322) };
\addlegendentry{algo=viecut128};
\addplot coordinates { (32768.0,12.4749) (65536.0,9.69447) (131072.0,11.4265) (262144.0,10.7501) (524288.0,10.5735) };
\addlegendentry{algo=viecutfast};
\addplot coordinates { (32768.0,8.91762) (65536.0,7.81309) (131072.0,7.89549) (262144.0,7.43826) (524288.0,7.2994) };
\addlegendentry{algo=viecut};
\addplot coordinates { (32768.0,8.5456) (65536.0,4.84844) (131072.0,5.86085) (262144.0,4.4461) (524288.0,4.3165) };
\addlegendentry{algo=parviecut128};
\addplot coordinates { (32768.0,4.85621) (65536.0,3.8601) (131072.0,4.77678) (262144.0,3.69705) (524288.0,3.6626) };
\addlegendentry{algo=parviecut};

\legend{}
\end{axis}

\begin{axis}[
  name=plot3,
  at=(plot2.right of south east),
  xshift=6mm,
    title={Average Node Degree: $2^7$},
    xlabel={Number of Vertices},
    ymode=log,
    log basis y={2},
    xmode=log,
    log basis x={2},
    legend pos=outer north east,
    legend style={font=\Large},
    label style={font=\Large},
    title style={font=\Large},
    tick label style={font=\Large},
    ]
\addplot coordinates { (32768.0,141.439) (65536.0,130.44) (131072.0,137.005) (262144.0,129.858) (524288.0,132.278) };
\addlegendentry{algo=ks\_ks\_nopr};
\addplot coordinates { (32768.0,20.6789) (65536.0,20.6081) (131072.0,19.465) (262144.0,20.0184) (524288.0,20.7504) };
\addlegendentry{algo=viecutglobal};
\addplot coordinates { (32768.0,15.5394) (65536.0,14.4336) (131072.0,14.1477) (262144.0,13.6153) (524288.0,13.6494) };
\addlegendentry{algo=matula};
\addplot coordinates { (32768.0,15.7411) (65536.0,14.7562) (131072.0,14.1816) (262144.0,14.1775) (524288.0,14.2447) };
\addlegendentry{algo=viecut1024};
\addplot coordinates { (32768.0,15.6941) (65536.0,14.9424) (131072.0,13.9666) (262144.0,14.2007) (524288.0,14.2246) };
\addlegendentry{algo=viecut512};
\addplot coordinates { (32768.0,15.7455) (65536.0,14.608) (131072.0,12.9032) (262144.0,13.62) (524288.0,13.6018) };
\addlegendentry{algo=viecut256};
\addplot coordinates { (32768.0,12.7863) (65536.0,11.4171) (131072.0,11.6398) (262144.0,11.5037) (524288.0,11.0658) };
\addlegendentry{algo=viecut128};
\addplot coordinates { (32768.0,10.0938) (65536.0,9.8202) (131072.0,9.46502) (262144.0,9.47371) (524288.0,9.40626) };
\addlegendentry{algo=viecutfast};
\addplot coordinates { (32768.0,7.59634) (65536.0,7.82408) (131072.0,7.03717) (262144.0,6.93932) (524288.0,6.95075) };
\addlegendentry{algo=viecut};
\addplot coordinates { (32768.0,4.89874) (65536.0,4.71355) (131072.0,4.24911) (262144.0,3.97356) (524288.0,3.89375) };
\addlegendentry{algo=parviecut128};
\addplot coordinates { (32768.0,3.95658) (65536.0,4.14293) (131072.0,3.47006) (262144.0,3.18223) (524288.0,3.32571) };
\addlegendentry{algo=parviecut};
\legend{}
\end{axis}

 \begin{axis}[
   name=plot4,
   at=(plot3.right of south east),
   xshift=6mm,
   title={Average Node Degree: $2^8$},
   xlabel={Number of Vertices},
   ymode=log,
   log basis y={2},
   xmode=log,
   log basis x={2},
   legend pos=outer north east,
   legend style={font=\Large},
   label style={font=\Large},
   title style={font=\Large},
   tick label style={font=\Large},
   ]
\addplot coordinates { (32768.0,150.319) (65536.0,149.701) (131072.0,145.411) (262144.0,133.959) (524288.0,138.948) };
\addlegendentry{\texttt{KS}};
\addplot coordinates { (32768.0,16.7182) (65536.0,15.229) (131072.0,16.0594) (262144.0,19.9515) (524288.0,16.7299) };
\addlegendentry{\texttt{\algname\_global}};
\addplot coordinates { (32768.0,13.6776) (65536.0,12.7918) (131072.0,12.7317) (262144.0,13.8367) (524288.0,12.9776) };
\addlegendentry{\texttt{Matula}};
\addplot coordinates { (32768.0,12.7836) (65536.0,11.4334) (131072.0,11.6135) (262144.0,13.7952) (524288.0,11.587) };
\addlegendentry{\texttt{\algname1024}};
\addplot coordinates { (32768.0,12.5076) (65536.0,11.4536) (131072.0,11.5207) (262144.0,13.9561) (524288.0,11.5852) };
\addlegendentry{\texttt{\algname512}};
\addplot coordinates { (32768.0,10.3827) (65536.0,10.2298) (131072.0,9.34389) (262144.0,12.4536) (524288.0,9.47789) };
\addlegendentry{\texttt{\algname256}};
\addplot coordinates { (32768.0,9.60362) (65536.0,9.04549) (131072.0,8.74115) (262144.0,11.4303) (524288.0,8.82937) };
\addlegendentry{\texttt{\algname128}};
\addplot coordinates { (32768.0,8.35007) (65536.0,7.92396) (131072.0,7.89582) (262144.0,9.27227) (524288.0,7.89209) };
\addlegendentry{\texttt{\algname\_fast}};
\addplot coordinates { (32768.0,6.95065) (65536.0,6.79874) (131072.0,6.64843) (262144.0,7.02401) (524288.0,6.57993) };
\addlegendentry{\texttt{\algname\_cons}};
\addplot coordinates { (32768.0,3.60322) (65536.0,3.6) (131072.0,3.34932) (262144.0,3.92536) (524288.0,3.32971) };
\addlegendentry{\texttt{par\algname128}};
\addplot coordinates { (32768.0,3.41214) (65536.0,3.0395) (131072.0,2.72394) (262144.0,3.3286) (524288.0,3.12961) };
\addlegendentry{\texttt{par\algname\_cons}};
\end{axis}
\end{tikzpicture}
}
\caption{Total running time in nanoseconds per edge in RHG graphs \label{fig:rhg}}
\end{figure}
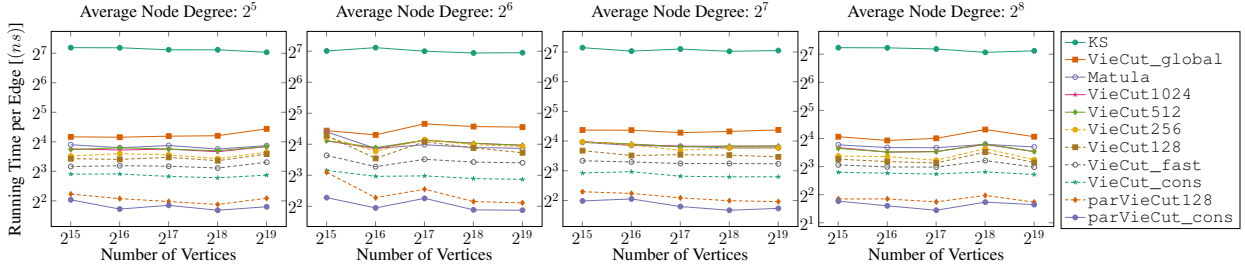%
 Figure~\ref{fig:rhg} shows the total running time for different configurations of \texttt{\algname}. From the sequential
 algorithms, \texttt{\algname\_cons} has the lowest running time for all algorithms. The
algorithm, however, returns non-optimal cuts in more than $\frac{1}{3}$ of all instances, with an
average distance to the minimum cut of $~44\%$ over all graphs. The best results were obtained by \texttt{\algname128},
which has an average distance of $0.83\%$ and only $10$ non-optimal results out of $300$ instances. The results are very good compared to \texttt{Matula}, which has $57$ non-optimal results in these $300$ instances and an average distance of $5.57\%$. \texttt{\algname128} is $20\%$ faster on most graphs than \texttt{Matula}, regardless of graph size or density.  In the main text
we use the configuration \texttt{\algname128} with $2$ iterations, there named
\texttt{\algname}. On these small graphs, the parallel versions have
a speedup factor of $2$ to $3.5$ compared to their sequential version. \texttt{par\algname128} has $17$ non-optimal results and an average distance of $4.91\%$ while \texttt{par\algname\_cons} has 29 non-optimal results and $20\%$ average distance to the minimum cut. Therefore we use \texttt{par\algname128} for all parallel experiments in the main text (named \texttt{par\algname}).

\vfill
\pagebreak
\section{Additional Figures}
\begin{figure}

\centering
\resizebox{.45\textwidth}{!}{
\begin{tikzpicture}
  \begin{axis}[
    name=plot1,
    ylabel={Running Time},
    height=6.5cm,
    width=4cm,
    xmode=log,
    ymin=0,
    xtick=\empty,
    stack plots=y,
    legend style={at={(3.49,0.99)}},
    ybar stacked,
    bar width=3cm,
    ytick={0,0.20,0.40,0.60,0.80,1},
    yticklabels={$0$,$20\%$,$40\%$,$60\%$,$80\%$,$100\%$},
    reverse legend,
    ymax=1,
    ]
    \definecolor{tempcolor1}{RGB}{27,158,119}
    \definecolor{tempcolor2}{RGB}{217,95,2}
    \definecolor{tempcolor3}{RGB}{117,112,179}
    \definecolor{tempcolor4}{RGB}{231,41,138}
    \definecolor{tempcolor5}{RGB}{102,166,30}
    \addplot [fill=tempcolor1, postaction={pattern=north east lines}] coordinates {(1, 0.855)};
    \addlegendentry{Label Propagation};
    \addplot [fill=tempcolor2, postaction={pattern=dots}] coordinates {(1, 0.011)};
    \addlegendentry{Fix Misplaced Vertices};
    \addplot [fill=tempcolor3, postaction={pattern=north west lines}] coordinates {(1, 0.134)};
    \addlegendentry{Graph Contraction};
    \addplot [fill=tempcolor4, postaction={pattern=vertical lines}] coordinates {(1, 0.000001)};
    \addlegendentry{Padberg-Rinaldi Tests};
    \addplot [fill=tempcolor5, postaction={pattern=horizontal lines}] coordinates {(1, 0.000001)};
    \addlegendentry{\texttt{NOI} Algorithm};
  \end{axis}
\end{tikzpicture}
}
\caption{Running Time Breakdown for RHG Graphs with $n=2^{25}$ and $m=2^{32}$\label{fig:breakdown}}
\end{figure}
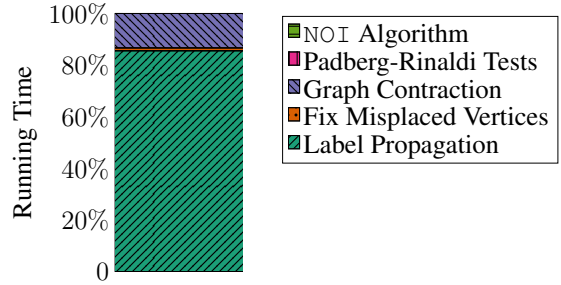
\begin{figure}[h]
\centering
  \resizebox{.7\textwidth}{!}{
\begin{tikzpicture}
  \begin{axis}[
    name=plot1,
    ylabel={Running Time Percentage},
    xlabel={Number of Threads},
    legend pos=outer north east,
    height=6.5cm,
    width=13cm,
    ymin=0,
    xtick=data,
    stack plots=y,
    ybar stacked,
    bar width=0.4cm,
    ytick={0,0.20,0.40,0.60,0.80,1},
    yticklabels={$0$,$20\%$,$40\%$,$60\%$,$80\%$,$100\%$},
    xticklabels={$1$,$2$,$3$,$4$,$5$,$6$,$7$,$8$,$9$,$10$,$11$,$12$,$24$},
    reverse legend,
    ymax=1,
    ]
    \definecolor{tempcolor1}{RGB}{27,158,119}
    \definecolor{tempcolor2}{RGB}{217,95,2}
    \definecolor{tempcolor3}{RGB}{117,112,179}
    \definecolor{tempcolor4}{RGB}{231,41,138}
    \definecolor{tempcolor5}{RGB}{102,166,30}
    \addplot[fill=tempcolor1, postaction={pattern=north east lines}] coordinates {(1,0.81) (2,0.78) (3,0.77) (4,0.77) (5,0.77) (6,0.76) (7,0.75) (8, 0.75) (9, 0.74) (10, 0.73) (11, 0.73) (12,0.73) (13,0.71)};
    \addlegendentry{Label Propagation};
    \addplot[fill=tempcolor2, postaction={pattern=dots}] coordinates {(1,0.001) (2,0.001) (3,0.001) (4,0.001) (5,0.001) (6,0.001) (7,0.001) (8, 0.001) (9, 0.001) (10, 0.001) (11, 0.001) (12,0.001) (13,0.001)};
    \addlegendentry{Fix Misplaced Vertices};
    \addplot[fill=tempcolor3, postaction={pattern=north west lines}] coordinates {(1,0.19) (2,0.22) (3,0.23) (4,0.23) (5,0.23) (6,0.24) (7,0.25) (8, 0.25) (9, 0.26) (10, 0.27) (11, 0.27) (12,0.27) (13,0.29)};
    \addlegendentry{Graph Contraction};
    \addplot[fill=tempcolor4, postaction={pattern=vertical lines}] coordinates {(1,0) (2,0) (3,0) (4,0) (5,0) (6,0) (7,0) (8, 0) (9, 0) (10, 0) (11, 0) (12,0) (13,0)};
    \addlegendentry{Padberg-Rinaldi Tests};
    \addplot[fill=tempcolor5, postaction={pattern=horizontal lines}] coordinates {(1,0) (2,0) (3,0) (4,0) (5,0) (6,0) (7,0) (8, 0) (9, 0) (10, 0) (11, 0) (12,0) (13,0)};
    \addlegendentry{\texttt{NOI} Algorithm};
  \end{axis}
\end{tikzpicture}
}
\caption{Parallel Running Time Breakdown in Large Graphs\label{fig:scaledown}}
\end{figure}
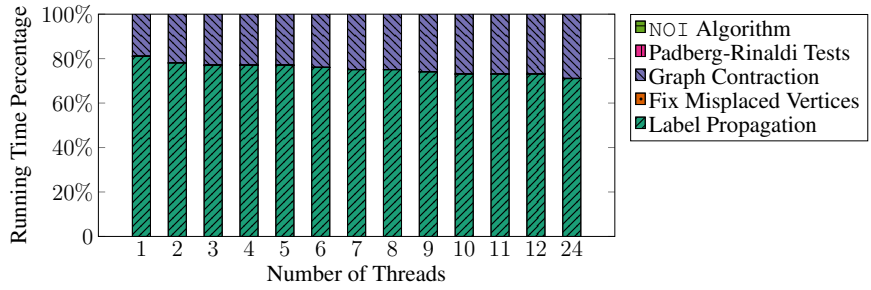
\end{appendix}
\end{document}